\documentclass{vldb}

\usepackage{xcolor}
\usepackage[hyphens]{url} 
\usepackage{amsmath}
\usepackage{bbm}

\usepackage{mathtools}
\usepackage{enumitem}
\usepackage{changepage} 
\usepackage{array}
\usepackage{multirow}
\usepackage{epstopdf}

\usepackage{booktabs}

\usepackage{etoolbox}

\newcommand{\rv}[1]{\mathbf{#1}}    
\newcommand{\prsmall}[1]{\mathbb{P}\mathrm{r}\{#1\}} 
\newcommand{\nn}{\nonumber}
\newcommand{\ds}[1]{\mathbb{D}_{\mathsf{#1}}}
\newcommand{\bu}[1]{\mathrm{bkt}(#1)}
\newcommand{\given}{\,|\,}
\newcommand{\pa}{P_\mathcal{G}}
\newcommand{\padp}{P_{\tilde{\mathcal{G}}}}
\newcommand{\pr}[1]{\mathbb{P}\mathrm{r}\{#1\}} 

\newcommand{\ignore}[1]{}

\newtheorem{mechanism}{\bf Mechanism}
\newtheorem{privacytest}{\bf Privacy Test}

\newtheorem{theorem}{\bf Theorem}
\newtheorem{theoremrec}{\bf Theorem}
\newtheorem{corollary}{\bf Corollary}
\newtheorem{lemma}{\bf Lemma}
\newtheorem{definition}{\bf Definition}
\newtheorem{observation}{\bf Observation}

\makeatletter
\def\@copyrightspace{\relax}
\makeatother

\begin{document}

\title{Plausible Deniability for Privacy-Preserving Data Synthesis (Extended Version)}

\numberofauthors{3}
\author{
\alignauthor
Vincent Bindschaedler\\
\affaddr{UIUC}\\
\email{bindsch2@illinois.edu}
\alignauthor
Reza Shokri\\
\affaddr{Cornell Tech}\\
\email{shokri@cornell.edu}
\alignauthor
Carl A. Gunter\\
\affaddr{UIUC}\\
\email{cgunter@illinois.edu}
}

\maketitle

\begin{abstract}

Releasing full data records is one of the most challenging problems in
data privacy. On the one hand, many of the popular techniques such as data
de-identification are problematic because of their dependence
on the background knowledge of adversaries.  On the other hand,
rigorous methods such as the exponential mechanism for differential
privacy are often computationally impractical to use for releasing high
dimensional data or cannot preserve high utility of original
data due to their extensive data perturbation.

This paper presents a criterion called \emph{plausible
deniability} that provides a formal privacy guarantee, notably for
releasing sensitive datasets: an output record can be released only
if a certain amount of input records are indistinguishable, up to a
privacy parameter. This notion does not depend on the background
knowledge of an adversary. Also, it can efficiently be checked by privacy tests. We present mechanisms to generate \emph{synthetic datasets} with similar statistical properties to the input data and the same format. We study this technique both theoretically and experimentally. A key theoretical result shows that, with proper randomization, the plausible deniability mechanism generates differentially private synthetic data.
We demonstrate the efficiency of this generative technique on a large dataset; it is shown to preserve the utility of original data with respect to various statistical analysis and machine learning measures.

\end{abstract}
\section{Introduction}

There is tremendous interest in releasing datasets for research and
development.  Privacy policies of data holders, however, prevent them
from sharing their sensitive datasets.  This is due, to a large
extent, to multiple failed attempts of releasing datasets using
imperfect privacy-preserving mechanisms such as de-identification. A
range of inference attacks on, for example, AOL search log dataset
\cite{nytaol2006}, Netflix movie rating
dataset~\cite{narayanan2008robust}, Genomic
data~\cite{wang2009learning,humbert2013addressing}, location
data~\cite{golle2009anonymity,privacynyctaxi}, and social networks
data \cite{narayanan2009anonymizing}, shows that simple modification
of sensitive data by removing identifiers or by generalizing/suppressing
data features results in major information leakage and cannot
guarantee meaningful privacy for data owners.  These simple
de-identification solutions, however, preserve data utility as they
impose minimal perturbation to real data.

Rigorous privacy definitions, such as differential privacy
\cite{dwork2006calibrating}, can theoretically guarantee privacy and
bound information leakage about sensitive data.  However, known
mechanisms, such as the Laplacian mechanism \cite{dwork2006calibrating} or the
exponential mechanism \cite{mcsherry2007mechanism}, that achieve
differential privacy through randomization, have practical
limitations.  The majority of scenarios, where they have been applied,
are limited to interactive count queries on statistical databases
\cite{dwork2008differential}.  In a non-interactive setting for
releasing generic datasets, these mechanisms are either
computationally infeasible on {\em high-dimensional} data, or
practically ineffective because of their large {\em utility} costs
\cite{kifer2011no}.  At best, these methods are used to
release some privacy-preserving statistics (e.g., histograms
\cite{blocki2016differentially, xu2013differentially}) about a
dataset, but not {\em full} data records.  It is not
obvious how to protect the privacy of full records as opposed to that
of aggregate statistics (by adding random noise).

Despite all these obstacles, releasing full data records is firmly
pursued by large-scale data holders such as the U.S. Census
Bureau~\cite{hawala2008producing, kinney2011towards, kinney2014synlbd}.
The purpose of this endeavor is to allow researchers to develop
analytic techniques by processing full synthetic data records rather
than a limited set of statistics.  Synthetic data could also be used
for educational purpose, application development for data analysis,
sharing sensitive data among different departments in a company,
developing and testing pattern recognition and machine learning
models, and algorithm design for sensitive data.  There exists some
inference-based techniques to {\em assess} the privacy risks of
releasing synthetic data \cite{reiter2009estimating,
reiter2014bayesian}.  However, the major open problem is how to {\em
generate} synthetic full data records with {\em provable privacy},
that experimentally can achieve acceptable utility in various
statistical analytics and machine learning settings.

In this paper, we fill this major gap in data privacy by proposing a
generic theoretical framework for generating synthetic data in a
privacy-preserving manner.  The fundamental difference between our
approach and that of existing mechanisms for differential privacy
(e.g., exponential mechanism) is that we disentangle the data
generative model from privacy definitions.  Instead of forcing a
generative model to be privacy-preserving by design, which might
significantly degrade its utility, we can use a utility-preserving
generative model and release only a subset of its output that
satisfies our privacy requirements.  Thus, for designing a generative
model, we rely on the state-of-the-art techniques from data science
independently from the privacy requirements.  This enables us to
generate high utility synthetic data.

We formalize the notion of {\em plausible deniability} for data
privacy~\cite{bindschaedler2016synthesizing}, and generalize it to
any type of data.  Consider a probabilistic generative model that
transforms a real data record, as its seed, into a synthetic data
record.  We can sample many synthetic data records from each seed
using such a generative model.  According to our definition, a
synthetic record provides plausible deniability if there exists a set
of real data records that could have generated the same synthetic data
with (more or less) the same probability by which it was generated
from its own seed.  We design a privacy mechanism that provably
guarantees plausible deniability.  This mechanism results in
\emph{input indistinguishability}: by observing the output set (i.e.,
synthetics), an adversary cannot tell for sure whether a particular
data record was in the input set (i.e., real data).  The degree of
this indistinguishability is a parameter in our mechanism.

Plausible deniability is a property of the overall process, and
similar to differential privacy, it is independent of any adversary's
background knowledge.  In fact, we prove that our proposed plausible deniable data synthesis process
can also satisfy differential privacy, if we randomize the indistinguishability parameter in the
privacy mechanism.  This is a significant theoretical result towards
achieving strong privacy using privacy-agnostic utility-preserving
generative models.  Thus, we achieve differential privacy {\em
without} artificially downgrading the utility of the synthesized data
through output perturbation.

The process of generating a single synthetic data record and testing
its plausible deniability can be done independently from that of other
data records.  Thus, millions of data records can be generated and
processed in parallel.  This makes our framework extremely efficient
and allows implementing it at a large scale.  In this paper, we
develop our theoretical framework as an open-source tool, and run it
on a large dataset: the American Community Survey~\cite{acsweb}
from the U.S. Census Bureau which contains over 3.1 million records. In fact, we can generate over one million
privacy-preserving synthetic records in less than one hour on a
multi-core machine running $12$ processes in parallel. 

We analyze the utility of synthetic data in two major scenarios:
extracting statistics for data analysis, and performing prediction
using machine learning.  We show that our privacy test does not impose
high utility cost.  We also demonstrate that a significant fraction of
candidate synthetic records proposed by a generative model can pass
the privacy test even for strict privacy parameters.

We show that a strong adversary cannot distinguish a synthetic record from a real one with better than 63.0\% accuracy (baseline: 79.8\%). Furthermore, when it comes to classification tasks, the accuracy of the model learned on a synthetic dataset is only slightly lower than that of model trained on real data. For example, for Random Forest the accuracy is 75.3\% compared to 80.4\% when trained on real data (baseline: 63.8\%); whereas for AdaBoostM1 the accuracy is 78.1\% compared to 79.3\% when trained on real data (baseline: 69.2\%). Similar results are obtained when we compare logistic regression (LR) and support vector machine (SVM) classifiers trained on our synthetic datasets with the same classifiers trained (on real data) in a differential private way (using state-of-the-art techniques). 
Concretely, the accuracy of classifiers trained on our synthetic data is $77.5\%$ (LR) and $77.1\%$ (SVM); compared to $76.3\%$ (LR) and $78.2\%$ (SVM) for objective-perturbation $\varepsilon$-DP classifiers.

\smallskip

\noindent {\bf Contributions.}  We introduce a formal framework for
plausible deniability as a privacy definition.  We also design a
mechanism to achieve it for the case of generating synthetic data.  We
prove that using a randomized test in our plausible deniability mechanism achieves
differential privacy (which is a stronger guarantee).  We also show
how to construct generative models with differential privacy guarantees.
The composition of our generative model and plausible deniability
mechanism also satisfies differential privacy.  We show the high accuracy
of our model and utility of our generated synthetic data.  We develop
a generic tool and show its high efficiency for generating millions of
full data records.
\section{Plausible Deniability}\label{sec:pd}
In this section, we formalize plausible deniability as a new privacy notion for releasing privacy-preserving synthetic data.  We also present a mechanism to achieve it.  Finally, we prove that our mechanism can also satisfy differential privacy (which is a stronger guarantee) by slightly randomizing our plausible deniability mechanism. 

Informally, plausible deniability states that an adversary (with any background knowledge) cannot deduce that a particular record in the input (real) dataset was significantly more responsible for an observed output (synthetic record) than was a collection of other input records.  A mechanism ensures plausible deniability if, for a privacy parameter $k>0$, there are at least $k$ input records that could have generated the observed output with similar probability. 

Unlike the majority of existing approaches (e.g., to achieve differential privacy), designing a mechanism to satisfy plausible deniability for generative models does not require adding artificial noise to the generated data.  Instead, we separate the process of releasing privacy-preserving data into running two independent modules: (1) generative models, and (2) privacy test.  The first consists in constructing a utility-preserving generative data model.  This is ultimately a data science task which requires insight into the type of data for which one wants to generate synthetics.  By contrast, the privacy test aims to safeguard the privacy of those individuals whose data records are in the input dataset.  Every generated synthetic is subjected to this privacy test; if it passes the test it can be safely released, otherwise it is discarded.  This is where the plausible deniability criterion comes into the frame: the privacy test is designed to ensure that any released output can be plausibly denied.

In this section, we assume a generic generative model that, given a data record in the input dataset as seed, produces a synthetic data record.  In Section~\ref{sec:gm}, we present a generic generative model based on statistical models, and show how it can be constructed in a differentially-private manner, so that it does not significantly leak about its own training data.  Plausibly deniable mechanisms protect the privacy of the seeds, and are not concerned about how the generative models are constructed. 

Let $\mathcal{M}$ be a probabilistic generative model that given any data record $d$ can generate synthetic records $y$ with probability $\pr{ y = \mathcal{M}(d) }$. Let $k \ge 1$ be an integer and $\gamma \ge 1$ be a real number. Both $k$ and $\gamma$ are privacy parameters. 
\begin{definition}[Plausible Deniability]\ \\
For any dataset $D$ with $|D| \ge k$, and any record $y$ generated by a probabilistic generative model $\mathcal{M}$ such that $y = \mathcal{M}(d_1)$ for $d_1 \in D$, we state that $y$ is releasable with $(k,\gamma)$-{\em plausible deniability}, if there exist at least $k-1$ distinct records $d_2, ..., d_k \in D \setminus \{ d_1 \}$ such that
	\begin{align} 
		\gamma^{-1} \leq \frac{\pr{y = \mathcal{M}(d_i)}}{\pr{y = \mathcal{M}(d_j)}} \leq \gamma ,
		\label{eq:pdcond}
	\end{align}
	for any $i,j \in \{1, 2, \ldots, k\}$.
\label{def:pdmech}
\label{def:pdcriterion}
\end{definition}
The larger privacy parameter $k$ is, the larger the indistinguishability set for the input data record.  Also, the closer to $1$ privacy parameter $\gamma$ is, the stronger the indistinguishability of the input record among other plausible records. 

Given a generative model $\mathcal{M}$, and a dataset $D$, we need a mechanism $\mathcal{F}$ to guarantee that the privacy criterion is satisfied for any released data. Specifically $\mathcal{F}$ produces data records by using $\mathcal{M}$ on dataset $D$. The following mechanism enforces $(k,\gamma)$-plausible deniability by construction. 
\begin{mechanism}[$\mathcal{F}$ with Plausible Deniability]\ \\ 
Given a generative model $\mathcal{M}$, dataset $D$, and parameters $k$, $\gamma$, output a synthetic record $y$ or nothing.
\begin{enumerate}[noitemsep]
	\item{Randomly sample a {\em seed} record $d \in D$.}
	\item{Generate a {\em candidate} synthetic record $y = \mathcal{M}(d)$.}
	\item{Invoke the {\bf privacy test} on $(\mathcal{M},D,d,y,k,\gamma)$.}
	\item{If the tuple passes the test, then release $y$. \\Otherwise, there is no output.}
\end{enumerate}
\label{def:mechanism}
\end{mechanism}

The core of Mechanism~\ref{def:mechanism} ($\mathcal{F}$) is a privacy test that simply rejects a candidate synthetic data record if it does not satisfy a given privacy criterion.

We can think of Definition~\ref{def:pdmech} as a \emph{privacy criterion} that can be efficiently checked and enforced. So, instead of trying to measure how sensitive the model $\mathcal{M}$ is with respect to input data records, we test if there are enough indistinguishable records in the input dataset that could have (plausibly) generated a candidate synthetic data record. 
\begin{privacytest}[Deterministic test $\mathcal{T}$]\ \\
Given a generative model $\mathcal{M}$, dataset $D$, data records $d$ and $y$, and privacy parameters $k$ and $\gamma$, output $\mathrm{pass}$ to allow releasing $y$, otherwise output $\mathrm{fail}$. 
\begin{enumerate}[noitemsep]
	\item{Let $i \ge 0$ be the (only) integer that fits the inequalities \[ \gamma^{-i-1} < \prsmall{y = \mathcal{M}(d)} \leq \gamma^{-i} \ . \]}
	\item{Let $k'$ be the number of records $d_a \in D$ such that \[ \gamma^{-i-1} < \prsmall{y = \mathcal{M}(d_a)} \leq \gamma^{-i} \ . \]}
	\item{If $k' \ge k$ then return $\mathrm{pass}$, otherwise return $\mathrm{fail}$}.
\end{enumerate}
\label{def:privacytest}
\end{privacytest}

Step 2 counts the number of \emph{plausible} seeds, i.e., records in $D$ which could have plausibly produced $y$. Note that for a given $y$, there may exist some records $d_a \in D$ such that $\prsmall{y = \mathcal{M}(d_a)} = 0$. Such records cannot be plausible seeds of $y$ since no integer $i \geq 0$ fits the inequalities.

Remark that Privacy Test~\ref{def:privacytest} ($\mathcal{T}$) enforces a stringent condition that the probability of generating a candidate synthetic $y$ given the seed $d$ and the probability of generating the same record given another plausible seed $d_a$ both fall into a geometric range $[\gamma^{-i-1},\gamma^{-i}]$, for some integer $i \ge 0$, assuming $\gamma > 1$. Notice that, under this test, the set of $k-1$ different $d_a$s plus $d$ satisfies the plausible deniability condition \eqref{eq:pdcond}.

Informally, the threshold $k$ prevents releasing the \emph{implausible} synthetics records $y$.  As $k$ increases the number of plausible records which could have produced $y$ also increases.  Thus, an adversary with only partial knowledge of the input dataset cannot readily determine whether a particular input record $d$ was the seed of any released record $y$. This is because there are at least $k-1$ other records $d_i \neq d$ in the input dataset which could \emph{plausibly} have been the seed. 
However, whether $y$ passes the privacy test itself reveals something about the number of plausible seeds, which could potentially reveal whether a particular $d$ is included in the input data.
This can be prevented by using a privacy test which randomizes the threshold $k$ (as Section~\ref{sec:pd:proofs} shows) in which case the mechanism achieves $(\varepsilon,\delta)$-differential privacy.

\subsection{Relationship with Differential Privacy}\label{sec:pd:proofs}

We show a connection between Plausible Deniability and Differential Privacy, given the following definition.
\begin{definition}[Differential Privacy~\cite{dwork2014algorithmic}]\ \\
Mechanism $F$ satisfies $(\varepsilon,\delta)$-\emph{differential privacy} if for any neighboring datasets $D$, $D'$, and any output $S \subseteq {\rm{Range}}(F)$:
	\[ \pr{F(D') \in S} \leq e^\varepsilon \pr{F(D) \in S} + \delta \ . \]
\end{definition}
Typically, one chooses $\delta$ smaller than an inverse polynomial in the size of the dataset, e.g., $\delta \leq |D|^{-c}$, for some $c > 1$.

In this section, we prove that if the privacy test is randomized in a certain way, then Mechanism~\ref{def:mechanism} ($\mathcal{F}$) is in fact $(\varepsilon,\delta)$-differentially private for some $\delta > 0$ and $\varepsilon > 0$.  Privacy Test~\ref{def:privacytest} simply counts the number of plausible seeds for an output and only releases a candidate synthetic if that number is at least $k$.  We design Privacy Test~\ref{def:randprivacytest} which is identical except that it randomizes the threshold $k$.
\begin{privacytest}[Randomized test $\mathcal{T}_{\epsilon_0}$]\ \\
Given a generative model $\mathcal{M}$, dataset $D$, data records $d$ and $y$, privacy parameters $k$ and $\gamma$, and randomness parameter $\epsilon_0$, output $\mathrm{pass}$ to allow releasing $y$, otherwise output $\mathrm{fail}$.
\begin{enumerate}[noitemsep]
	\item{Randomize $k$ by adding fresh noise: $\tilde{k} = k + \mathrm{Lap}(\frac{1}{\epsilon_0})$}.
	\item{Let $i \geq 0$ be the (only) integer that fits the inequalities \[ \gamma^{-i-1} < \prsmall{y = \mathcal{M}(d)} \leq \gamma^{-i} \ . \]}
	\item{Let $k'$ be the number of records $d_a \in D$ such that \[ \gamma^{-i-1} < \prsmall{y = \mathcal{M}(d_a)} \leq \gamma^{-i} \ . \]}
	\item{If $k' \ge \tilde{k}$ then return $\mathrm{pass}$, otherwise return $\mathrm{fail}$}.
\end{enumerate}
\label{def:randprivacytest}
\end{privacytest}
Here $z \sim \mathrm{Lap}(b)$ is a sample from the Laplace distribution $\frac{1}{2b} \exp{(\frac{-|z|}{b})}$ with mean $0$ and shape parameter $b > 0$.

\begin{theorem}[Differential Privacy of $\mathcal{F}$]\ \\
Let $\mathcal{F}$ denote Mechanism~\ref{def:mechanism} with the (randomized) Privacy Test~\ref{def:randprivacytest} and parameters $k \geq 1$, $\gamma > 1$, and $\varepsilon_0 > 0$. For any neighboring datasets $D$ and $D'$ such that $|D|, |D'| \geq k$, any set of outcomes $Y \subseteq \mathcal{U}$, and any integer $1 \leq t < k$, we have:
	\begin{align*}
		\pr{\mathcal{F}(D') \in Y} \leq e^{\varepsilon} \pr{\mathcal{F}(D) \in Y} + \delta \ ,
	\end{align*}
	for $\delta = e^{-\varepsilon_0 (k-t)}$ and $\varepsilon = \varepsilon_0 + \ln{(1+\frac{\gamma}{t})}$.
	\label{thm:dp}
	\label{thm:dpconn}
\end{theorem}

The privacy level offered by Theorem~\ref{thm:dp} is meaningful provided $k$ is such that $\delta$ is sufficiently small. For example, if we want $\delta \leq \frac{1}{n^c}$ for some $c > 1$, then we can set $k \geq t + \frac{c}{\varepsilon_0} \ln{n}$. Here $t$ provides a trade-off between $\delta$ and $\varepsilon$.

The proof of Theorem~\ref{thm:dp} can be found in Appendix~\ref{app:connection}. 
Roughly speaking, the theorem says that, except with some small probability $\delta$, adding a record to a dataset cannot change the probability that any synthetic record $y$ is produced by more than a small multiplicative factor. The intuition behind this is the following.

Fix an arbitrary synthetic record $y$ produced by the mechanism on some dataset. Remark that given $y$, records are partitioned into disjoint sets according to their probabilities of generating $y$ (with respect to $\mathcal{M}$). That is, partition $i$ for $i=0,1,2\ldots$, contains those records $d$ such that $\gamma^{-(i+1)} < \pr{y = \mathcal{M}(d)} \leq \gamma^{-i}$. (We ignore records which have probability $0$ of generating $y$.)

Now suppose we add an arbitrary record $d'$ to the dataset. The probability of producing $y$ changes in two ways: (1) the probability that $y$ is generated increases because $d'$ may be chosen as seed, and (2) the probability that $y$ passes the privacy test increases because $d'$ is an additional plausible seed. Remark that this change only impacts whichever partition $d'$ falls into because the probability of passing the privacy test depends only on the number of plausible seeds in the partition of the seed.

Thus, we focus on the partition in which $d'$ falls. If that partition contains a small number of records compared to $k$ then introducing $d'$ could increase the probability of generating $y$ significantly, but the probability of passing the privacy test is very small. (The likelihood of passing the privacy test decreases exponentially the fewer plausible seeds are available compared to $k$.) In contrast, if the partition contains a number of records comparable to $k$ or larger, then the probability of generating $y$ increases only slightly (because there are already a large number of plausible seeds with similar probability of generating $y$ as $d'$). And, the probability of passing the privacy test increases by a multiplicative factor of at most $e^{\varepsilon_0}$ due to adding Laplacian noise. In both cases, the increase to the probability of producing $y$ due to adding $d'$ is small and can be bounded.
\section{Generative Model}\label{sec:gm}

In this section, we present our generative model, and the process of using it to generate synthetic data.  The core of our synthesizer is a probabilistic model that captures the joint distribution of attributes.  We learn this model from training data samples drawn from our real dataset $\ds{}$.  Thus, the model itself needs to be privacy-preserving with respect to its training set.  We show how to achieve this with differential privacy guarantees.

Let $\ds{S}$, $\ds{T}$, and $\ds{P}$ be three non-overlapping subsets of dataset $\ds{}$.  We use these datasets in the process of synthesis, structure learning, and parameter learning, respectively. 

\subsection{Model}\label{sec:gm:model}

Let $\{\rv{x}_1, \rv{x}_2, ..., \rv{x}_m \}$ be the set of random variables associated with the attributes of the data records in $\ds{}$.  Let $\mathcal{G}$ be a directed acyclic graph (DAG), where the nodes are the random variables, and the edges represent the probabilistic dependency between them.  A directed edge from $\rv{x}_j$ to $\rv{x}_i$ indicates the probabilistic dependence of attribute $i$ to attribute $j$.  Let $\pa(i)$ be the set of parents of random variable $i$ according to the dependency graph $\mathcal{G}$.  The following model, which we use in Section~\ref{sec:gm:synthesis} to generate synthetic data, represents the joint probability of data attributes.
\begin{align}
	\pr{ \rv{x}_1, ..., \rv{x}_m } = 
	\prod_{i=1}^m 
	\pr{ \rv{x}_i \given \{ \rv{x}_j \}_{\forall{j \in \pa(i)}} }
\label{eq:gm:model}
\end{align} 

This model is based on a structure between random variables, captured by $\mathcal{G}$, and a set of parameters that construct the conditional probabilities.  In Section~\ref{sec:gm:structure} and Section~\ref{sec:gm:parameters}, we present our differentially-private algorithms to learn the structure and parameters of the model from $\ds{}$, respectively.

\subsection{Synthesis}\label{sec:gm:synthesis}

Using a generative model, we probabilistically transform a real data record (called the seed) into a synthetic data record, by updating its attributes.  Let $\{x_1, x_2, ..., x_m \}$ be the values for the set of data attributes for a randomly selected record in the seed dataset $\ds{S}$.  Let $\omega$ be the number of attributes for which we generate new values.  Thus, we keep (i.e., copy over) the values of $m - \omega$ attributes from the seed to the synthetic data.  Let $\sigma$ be a permutation over $\{1, 2, ..., m\}$ to determine the re-sampling order of attributes. 

We set the re-sampling order $\sigma$ to be the dependency order between random variables.  More precisely, $\forall{j \in \pa(i)}$: $\sigma(j) < \sigma(i)$.  We fix the values of the first $m - \omega$ attributes according to $\sigma$ (i.e., the synthetic record and the seed overlap on their $\{ \sigma(1), ..., \sigma(m-\omega) \}$ attributes).  We then generate a new value for each of the remaining $\omega$ attributes, using the conditional probabilities \eqref{eq:gm:model}.  As we update the record while we re-sample, each new value can depend on attributes with updated values as well as the ones with original (seed) values. 

We re-sample attribute $\sigma(i)$, for $i > m-\omega$, as
\begin{align}
	x'_{\sigma(i)} \sim 
	\pr{ \rv{x}_{\sigma(i)} \given 
	& \{ \rv{x}_{\sigma(j)} = x_{\sigma(j)} \}_{\forall{j \in \pa(i), j \leq m-\omega}}, \nn\\
	& \{ \rv{x}_{\sigma(j)} = x'_{\sigma(j)} \}_{\forall{j \in \pa(i), j > m-\omega}} 
	}
\label{eq:gm:seedbased_synthesis}
\end{align}

In Section~\ref{sec:pd}, we show how to protect the privacy of the {\em seed} data record using our plausible deniability mechanisms.

\medskip \noindent \textbf{Baseline: Marginal Synthesis.} As a baseline generative model, we consider a synthesizer that (independently from any seed record) samples a value for an attribute from its marginal distribution.  Thus, for all attribute $i$, we generate $x_i \sim \prsmall{\rv{x}_i}$.  This is based on an assumption of independence between attributes' random variables, i.e., it assumes $\prsmall{\rv{x}_1, ..., \rv{x}_m} = \prod_{i=1}^m \prsmall{\rv{x}_i}$.

\subsection{Privacy-Preserving Structure Learning}\label{sec:gm:structure}

Our generative model depends on the dependency structure between random variables that represent data attributes.  The dependency graph $\mathcal{G}$ embodies this structure.  In this section, we present an algorithm that learns $\mathcal{G}$ from real data, in a privacy-preserving manner such that $\mathcal{G}$ does not significantly depend on individual data records.  

The algorithm is based on maximizing a scoring function that reflects how correlated the attributes are according to the data.  There are multiple approaches to this problem in the literature~\cite{margaritis2003learning}.  We use a method based on a well-studied machine learning problem: {\em feature selection}.  For each attribute, the goal is to find the best set of features (among all attributes) to predict it, and add them as the attribute's parents, under the condition that the dependency graph remains acyclic. 

The machine learning literature proposes several ways to rank features in terms of how well they can predict a particular attribute.  One possibility is to calculate the information gain of each feature with the target attribute.  The major downside with this approach is that it ignores the redundancy in information between the features.  We propose to use a different approach, namely Correlation-based Feature Selection (CFS)~\cite{hall1999correlation} which consists in determining the best subset of predictive features according to some correlation measure.  This is an optimization problem to select a subset of features that have high correlation with the target attribute and at the same time have low correlation among themselves.  The task is to find the best subset of features which maximizes a merit score that captures our objective.

We follow~\cite{hall1999correlation} to compute the merit score for a parent set $\pa(i)$ for attribute $i$ as
\begin{align}
	\mathrm{score}(\pa(i)) = \frac{\sum_{j\in\pa(i)} \mathrm{corr}(\rv{x}_i, \rv{x}_j)}{\sqrt{|\pa(i)| +  \sum_{j,k\in\pa(i)} \mathrm{corr}(\rv{x}_j, \rv{x}_k)}} , 
\label{eq:gm:merit}
\end{align}
where $|\pa(i)|$ is the size of the parent set, and $\mathrm{corr}()$ is the correlation between two random variables associated with two attributes.  The numerator rewards correlation between parent attributes and the target attribute, and the denominator penalizes the inner-correlation among parent attributes.  The suggested correlation metric in~\cite{hall1999correlation}, which we use, is the symmetrical uncertainty coefficient:
\begin{align}
	\mathrm{corr}(\rv{x}_i, \rv{x}_j) = 2 - 2 \frac{\mathrm{H}(\rv{x}_i, \rv{x}_j)}{\mathrm{H}(\rv{x}_i) + \mathrm{H}(\rv{x}_j)} ,
\label{eq:gm:correlation}
\end{align}
where $\mathrm{H}()$ is the entropy function.

The optimization objective in constructing $\mathcal{G}$ is to maximize the total $\mathrm{score}(\pa(i))$ for all attributes $i$.  Unfortunately, the number of possible solutions to search for is exponential in the number of attributes, making it impractical to find the optimal solution.  The greedy algorithm, suggested in~\cite{hall1999correlation}, is to start with an empty parent set for a target attribute and always add the attribute (feature) that maximizes the score.

There are two constraints in our optimization problem.  First, the resulting dependency graph obtained from the set of best predictive features (i.e., parent attributes) for all attributes should be acyclic.  This would allow us to decompose and compute the joint distribution over attributes as represented in \eqref{eq:gm:model}.  

Second, we enforce a maximum allowable complexity cost for the set of parents for each attribute.  The cost is proportional to the number of possible joint value assignments (configurations) for the parent attributes.  So, for each attribute $i$, the complexity cost constraint is
\begin{align}
	\mathrm{cost}(\pa(i)) = \prod_{j\in\pa(i)} |\rv{x}_j| \leq \texttt{maxcost}
\label{eq:gm:constraint}
\end{align}
where $|\rv{x}_j|$ is the total number of possible values that the attribute $j$ takes.  This constraint prevents selecting too many parent attribute combinations for predicting an attribute.  The larger the joint cardinality of attribute $i$'s parents is, the fewer data points to estimate the conditional probability $\pr{ \rv{x}_i \given \{ \rv{x}_j \}_{\forall{j \in \pa(i)}} }$ can be found.  This would cause overfitting the conditional probabilities on the data, that results in low confidence parameter estimation in Section~\ref{sec:gm:parameters}.  The constraint prevents this.

To compute the score and cost functions, we discretize the parent attributes.  Let $\bu{}$ be a discretizing function that partitions an attribute's values into buckets.  If the attribute is continuous, it becomes discrete, and if it is already discrete, $\bu{}$ might reduce the number of its bins.  Thus, we update conditional probabilities as follows.
\begin{align}
	\pr{ \rv{x}_i \given \{ \rv{x}_j \}_{\forall_{j \in \pa(i)} } } 
	\approx 
	\pr{ \rv{x}_i \given \{ \bu{\rv{x}_j} \}_{\forall_{j \in \pa(i)}} }
\label{eq:gm:discretization}
\end{align}
where the discretization, of course, varies for each attribute.  We update \eqref{eq:gm:merit} and \eqref{eq:gm:constraint} according to \eqref{eq:gm:discretization}.  This approximation itself decreases the cost complexity of a parent set, and further prevents overfitting on the data.

\subsubsection{Differential-Privacy Protection}\label{sec:gm:structure_dp}

In this section, we show how to safeguard the privacy of individuals whose records are in $\ds{}$, and could influence the model structure (which might leak about their data).  

All the computations required for structured learning are reduced to computing the correlation metric \eqref{eq:gm:correlation} from $\ds{}$.  Thus, we can achieve differential privacy~\cite{dwork2014algorithmic} for the structure learning by simply adding appropriate noise to the metric.  As, the correlation metric is based on the entropy of a single or a pair of random variables, we only need to compute the entropy functions in a differentially-private way.  We also need to make sure that the correlation metric remains in the $[0, 1]$ range, after using noisy entropy values.

Let $\tilde{H}(\rv{z})$ be the noisy version of the entropy of a random variable $\rv{z}$, where in our case, $\rv{z}$ could be a single or pair of random variables associated with the attributes and their discretized version (as presented in \eqref{eq:gm:discretization}).  To be able to compute differentially-private correlation metric in all cases, we need to compute noisy entropy $\tilde{H}(\rv{x}_i)$, $\tilde{H}(\bu{\rv{x}_i})$, $\tilde{H}(\rv{x}_i, \rv{x}_j)$, and $\tilde{H}(\rv{x}_i, \bu{\rv{x}_j})$, for all attributes $i$ and $j$.  For each of these cases, we generate a {\em fresh} noise drawn from the Laplacian distribution and compute the differentially-private entropy as
\begin{align}
	\tilde{H}(\rv{z}) = \mathrm{H}(\rv{z}) + \mathrm{Lap}(\frac{\Delta_H}{\varepsilon_H})
\label{eq:gm:entropy}
\end{align}
where $\Delta_H$ is the sensitivity of the entropy function, and $\varepsilon_H$ is the differential privacy parameter.

It can be shown that if $\rv{z}$ is a random variable with a probability distribution, estimated from $n_T = | \ds{T} |$  data records, then the upper bound for the entropy sensitivity is
\begin{align}
	\Delta_{H} \leq \frac{1}{n_T} [2 + \frac{1}{\ln(2)} + 2 \log_2{n_T}] = O(\frac{\log_2{n_T}}{n_T})
\label{eq:gm:entropy_sensitivity}
\end{align}

The proof of~\eqref{eq:gm:entropy_sensitivity} can be found in Appendix~\ref{app:proofs}.
Remark that $\Delta_{H}$ is a function of $n_T$ (the number of records in $\ds{T}$) which per se needs to be protected.  As a defense, we compute $\Delta_{H}$ in a differentially-private manner, by once randomizing the number of records
\begin{align}
	\tilde{n}_T = n_T + \mathrm{Lap}(\frac{1}{\varepsilon_{n_T}})
\label{eq:gm:n_sensitivity}
\end{align}
 
By using the randomized entropy values, according to \eqref{eq:gm:entropy}, the model structure, which will be denoted by $\tilde{\mathcal{G}}$, is differentially private.  In Section~\ref{sec:gm:dprivacy}, we use the composition theorems to analyze the total privacy of our algorithm for obtaining a differentially-private structure. 

\subsection{Privacy-Preserving Parameter Learning}\label{sec:gm:parameters}
Having a dependency structure $\tilde{\mathcal{G}}$, we need to compute the conditional probabilities for predicting each of the attributes given its parent set (see \eqref{eq:gm:model}).  This is a well-known problem in statistics.  In this section, we show how to learn the parameters that represent such conditional probabilities, from $\ds{P}$, in a differentially private manner.

The problem to be solved is to first learn a prior distribution over the parameters of the conditional probabilities.  To do so, we learn the hyper-parameters (the parameters of the prior distribution over the model's parameters) from data.  Only then, we can compute the parameters that form the conditional probabilities from the prior distribution.

Let us take the example of computing the parameters for predicting discrete/categorical attributes.  In this case, we assume a multinomial distribution over the attribute's values (that fall into different bins).  The conjugate prior for multinomials comes from a Dirichlet family.  The Dirichlet distribution assigns probabilities to all possible multinomial distributions, according to the statistics obtained from a set of data records.  

Let $|\rv{x}_i|$ be the number of distinct values that attribute $i$ can take.  The probability of some multinomial distribution parameters $\vec{p\,}^{c}_i = p^{c}_{i,1}, p^{c}_{i,2}, ..., p^{c}_{i,|\rv{x}_i|}$ to predict attribute $i$, under configuration $c$ for $\pa(i)$, is
\begin{align}
	\pr{\vec{p\,}^{c}_i \given \tilde{\mathcal{G}}, \ds{P}} = 
	\mathrm{Dir}(\vec{\alpha\,}^c_i + \vec{n\,}^c_i)
\label{eq:gm:multinomial_probability}
\end{align}
where $\vec{\alpha\,}^c_i$ is the vector of default hyper-parameters for the Dirichlet distribution, and $\vec{n\,}^c_i$ is the vector for the number of data records in $\ds{P}$ with $\padp(i)$ configuration $c$ with different values for attribute $i$ (i.e., element ${n\,}^c_{i,l}$ is the number of records for which $x_i = l$ and $\padp(i)$ configuration is $c$).  The Dirichlet distribution is computed as
\begin{align}
	\mathrm{Dir}(\vec{\alpha\,}^c_i + \vec{n\,}^c_i) = 
	\prod_{i=1}^m \prod_{c=1}^{\#c} \Gamma(\alpha^c_i + n^c_i) \prod_{l=1}^{|\rv{x}_i|} \frac{(p^c_i)^{\alpha^c_{i,l} + n^c_{i,l}-1}}{\Gamma(\alpha^c_{i,l} + n^c_{i,l})}
\label{eq:gm:dirichlet}
\end{align}
where $\alpha^c_i = \sum_l \alpha^c_{i,l}$, and $n^c_i = \sum_l n^c_{i,l}$, and the number of configurations $\#c$ is $\prod_{j\in\padp(i)} |\rv{\bu{x_j}}|$, which according to constraint \eqref{eq:gm:constraint} can at most be $\texttt{maxcost}$.  

Learning the parameters of the model, in the case of a Dirichlet prior for multinomial distribution, is simply computing $\vec{n\,}^c_i$ from the data records in $\ds{P}$.  Given the probability distribution \eqref{eq:gm:multinomial_probability} over the multinomial parameters, we can compute the most likely set of parameters as 
\begin{align}
	p^c_{i,l} = \frac{\alpha^c_{i,l} + n^c_{i,l}}{\alpha^c_{i} + n^c_{i}}
\label{eq:gm:mlparameters}
\end{align}
or, we can sample a set of multinomial parameters according to \eqref{eq:gm:dirichlet}.  This is what we do in our generative model, in order to increase the variety of data samples that we can generate. 

Note that for computing the marginal distributions, that are needed for the baseline, we perform the same computations by setting the parent sets to be empty.

If an attribute is continuous, we can learn the parameters of a Normal distribution or learn a regression model from our data to construct its conditional probability.  We omit the details here (as in the dataset we evaluate in Section~\ref{sec:evaluation} all attributes are discrete).

\subsubsection{Differential-Privacy Protection}\label{sec:gm:parameters_dp}

The parameters of the conditional probabilities depend on the data records in $\ds{P}$, thus they can leak sensitive information about individuals who contributed to the real dataset.  In this section, we show how to learn parameters of the attribute conditional probabilities (i.e., $\vec{p\,}^{c}_i$ values) with differential privacy guarantees.

Note that in \eqref{eq:gm:multinomial_probability}, the only computations that are dependent on $\ds{P}$ are the $\vec{n\,}^c_i$ counts (for all $c$ and $i$).  To find the variance of the noise to be added to these counts, to achieve differential privacy, we need to compute their sensitivity with respect to one individual's data record.  

Suppose we are computing the parameters associated with predicting a given attribute $i$ given its parent set $\padp(i)$.  Note that adding a record to $\ds{P}$ increases exactly a single component ${n}^c_{i,l}$, for which it matches value $l$ for attribute $i$ and configuration $c$ for its parent set.  So, only one single element among all $\#c \times |\rv{x}_i|$ elements of $\vec{n}_i = \vec{n\,}^1_i, \vec{n\,}^2_i, ..., \vec{n\,}^{\#c}_i$ changes.  This implies that the $\mathrm{L1}$ sensitivity of $\vec{n}_i$ is $1$.  Consequently, a random noise drawn from $\mathrm{Lap}(\frac{1}{\varepsilon_p})$ can be added to each component of $\vec{n}_i$ independently.  More precisely, for any attribute $i$ value $l$, and configuration $c$, we randomize counts as
\begin{align}
	\tilde{n}^c_{i,l} = \max(0, n^c_{i,l} + \mathrm{Lap}(\frac{1}{\varepsilon_p}))
\end{align}
and use them to compute \eqref{eq:gm:multinomial_probability} with differential privacy.

\subsection{Differential Privacy Analysis}\label{sec:gm:dprivacy}

In this section, we compute the differential privacy level that we can guarantee for the whole dataset $\ds{}$ for learning the structure and the parameters of the model.  We compute the total $(\epsilon, \delta)$ privacy by composing the differentially-private mechanisms in Section~\ref{sec:gm:structure_dp} and Section~\ref{sec:gm:parameters_dp}.

Remark that we often protect the output $f(x)$ of some function $f$ by adding noise from $\mathrm{Lap}(\frac{\Delta_f}{\varepsilon})$, where $\Delta_f$ is the $\mathrm{L1}$ sensitivity of $f$. This mechanism is known as the Laplace mechanism and it satisfies $\varepsilon$-differential privacy (Theorem 3.6 of~\cite{dwork2014algorithmic}).

Thus the $m (m+1)$ entropy values, $\tilde{H}(\rv{z})$, needed for structure learning (Section~\ref{sec:gm:structure_dp}) are obtained in an a way that satisfies $\varepsilon_H$-differential privacy. This is also the case for the number of records $n_T$, i.e., it satisfies $\varepsilon_{n_T}$-differential privacy. Similarly, the counts $n^c_{i,l}$ parameters learned for each configuration (Section~\ref{sec:gm:parameters_dp}) satisfy $\varepsilon_p$-differential privacy.

For structure learning, we make use of both sequential composition (Theorem~\ref{thm:dp:seqc}) and advanced composition (Theorem~\ref{thm:dp:advc}).
Specifically, we use advanced composition for the $m (m+1)$ entropy values and sequential composition with the number of records. That is, the overall privacy achieved (of structure learning) is $(\varepsilon_L,\delta_L)$-differential privacy for a fixed $\delta_L \ll \frac{1}{n_T}$ and $\varepsilon_L = \varepsilon_{n_T} + \varepsilon_H \sqrt{2 m(m+1) \ln{(\delta_L^{-1})}}+ m(m+1) \varepsilon_H (e^{\varepsilon_H} - 1)$.

For the parameter learning (as explain in Section~\ref{sec:gm:parameters_dp}), for a given attribute $i$, the $\mathrm{L1}$ sensitivity of all configurations of the parent set of $i$, i.e., $\padp(i)$, is $1$. The overall privacy achieved (of parameter learning) is $(\varepsilon_P,\delta_P)$-differential privacy using advanced composition over the $m$ attributes. Here, $\delta_P \ll \frac{1}{n_p}$ (where $n_p$ is the number of records in $\ds{P}$) and $\varepsilon_P = \varepsilon_p \sqrt{2 m \ln{(\delta_P^{-1})}} + m \varepsilon_p (e^{\varepsilon_p} - 1)$.

Given that $\ds{T}$ and $\ds{P}$ are non-overlapping, the privacy obtained for the generative model is differentially private with parameters $(\max\{\varepsilon_L, \varepsilon_P\}, \max\{\delta_L,\delta_P\})$.  Due to random subsampling of $\ds{T}$ and $\ds{P}$ from $\ds{}$, the privacy parameters can be further improved by using the amplification effect of sampling (Theorem~\ref{thm:dp:amplif}) to obtain $(\varepsilon,\delta)$-differential privacy.
\section{Data}\label{sec:data}
For validation, we use the 2013 American Community Survey (ACS)~\cite{acsweb} from the U.S. Census Bureau. The dataset contains upwards of 3 million of individual records. Each record includes a variety of demographics attributes such as age, sex, race, as well as attributes related to the individual's income such as yearly income in USD.

The ACS dataset has been used for various purposes ranging from examining the relationship between education and earnings~\cite{julian2011education} to looking at current language use patterns of children of immigrants~\cite{mora2005language}. Furthermore, the prominent UCI Adult dataset, which provides a well-established benchmark for machine learning tasks, was extracted from the 1994 Census database. The 2013 ACS dataset contains similar attributes so we process it in a manner similar to how the Adult dataset was extracted. In particular, we extract the same attributes whenever possible.

As pre-processing, we discard records with missing or invalid values for the considered attributes (Table~\ref{tbl:acs13attrs}). 
Table~\ref{tbl:acs13stats} shows some statistics of the data cleaning and extracted dataset. This is a highly dimensional dataset despite having only $11$ attributes, there are more than half a trillion \emph{possible} records and out of the roughly $1.5$ million records obtained after cleaning, approximately $2/3$ are unique.
\begin{table}
    \caption{\small Pre-processed ACS13 dataset attributes.}
    \vspace{4pt}
	{\scriptsize
	\resizebox{\columnwidth}{!}{%
    \begin{tabular}{|l|l|l|}
    \hline
	    Name & Type & Cardinality (Values)         \\ \hline \hline
		Age (AGEP) & Numerical & $80$ ($17$ to $96$)         \\ \hline
		Workclass (COW) & Categorical & $8$        \\ \hline
		Education (SCHL) & Categorical & $24$         \\ \hline
		Martial Status (MAR) & Categorical & $5$       \\ \hline
		Occupation (OCCP) & Categorical & $25$         \\ \hline
		Relationship (RELP) & Categorical &$18$         \\ \hline
		Race (RAC1P) & Categorical & $5$         \\ \hline
		Sex (SEX) & Categorical & $2$ (male or female)         \\ \hline
		Hours Worked per Week (WKHP) & Numerical & $100$ ($0$ to $99$+)         \\ \hline
		World Area of Birth (WAOB) & Categorical & $8$     \\ \hline
		Income Class (WAGP) & Categorical & $2$ ($\leq 50$K, $> 50$K)[USD]         \\ \hline
    \end{tabular}
    }
    }
	\vspace{-6pt}
    \label{tbl:acs13attrs}
\end{table}

\begin{table}
    \caption{\small ACS13 data extraction and cleaning statistics.}
    \vspace{4pt}
	{\scriptsize	
	\resizebox{\columnwidth}{!}{%
    \begin{tabular}{|p{2.7cm}|p{4.9cm}|}
    \hline
	    Records                        & $3,132,796$ (clean: $1,494,974$)         \\ \hline
	    Attributes                     & $11$ (numerical: $2$, categorical: $9$)    \\ \hline
	    Possible Records 			   & $540,587,520,000$ ($\approx 2^{39}$)       \\ \hline
	    Unique Records                 & $1,022,718$ ($68.4\%$)                  \\ \hline
	    Classification Task            & Income class \\ \hline
    \end{tabular}
    }
    }
	\vspace{-6pt}
    \label{tbl:acs13stats}
\end{table}

We bucketize (Section~\ref{sec:gm:structure}) values of the age attribute in bins (i.e., buckets) of $10$, i.e., $17$ to $26$, $27$ to $36$, etc. (Following the rules used to extract the Adult dataset, we only consider individuals older than $16$.) We also bucketize the values of: hours worked per week (HPW), in bins of $15$ hours; education, to aggregate education level below a high-school diploma in a single bin, and high-school diploma but not college into (another) single-bin. Bucketization is performed based on the data format and the semantics of attributes (and thus is privacy-preserving). It is done only for structured learning (Section~\ref{sec:gm:structure}); both the input and output data format remain the same.
\section{Synthetics Generator Tool}\label{sec:tool}
The synthetic generator~\cite{tool} is implemented as C++ tool which takes as input: a dataset represented as a CSV file, a few metadata text files describing the dataset, and a config file. As output, the tool produces a synthetic dataset of the requested size and some metadata.

The generation process is defined by the config file, i.e., parameters defined within in control various aspects of the generation process. The parameters are the privacy parameters $k$, $\gamma$, $\varepsilon_0$, and also parameters of the generative model such as $\omega$. In addition, the tool takes two optional parameters to control the privacy test: \texttt{max\_plausible} and \texttt{max\_check\_plausible}, which allow the test to terminate early. Specifically, the implementation initially sets $k' = 0$, and iterates over the records of $D$ in a random order, incrementing $k'$ for each plausible seed record $d_a$ encountered. The process terminates whenever $k' \geq \texttt{max\_plausible}$ or if \texttt{max\_check\_plausible} records have been examined (whichever occurs first). Note that this affects performance (but not privacy); lower values lead to faster generation time in cases where plausible seeds are abundant, at the cost of fewer synthetics passing the test (potentially lowering utility).

The synthesis process, given a chosen seed, is independent of other seeds (Section~\ref{sec:pd}); so the generation process itself is embarrassingly parallel. One hurdle with running multiple concurrent instances is implementing differentially-private parameters learning (Section~\ref{sec:gm:parameters_dp}). In general, the number of configuration (in the sense of Section~\ref{sec:gm:parameters}) of the model is too large (i.e., exponential in the number of attributes) to learn the model as a pre-processing step. So we design the tool to learn the model for each configuration as it encounters it. To ensure that the privacy guarantee holds we set the RNG seed number to be a deterministic function (i.e., a hash) of the configuration.
\begin{figure}
	\centering
	\includegraphics[width=0.95\linewidth]{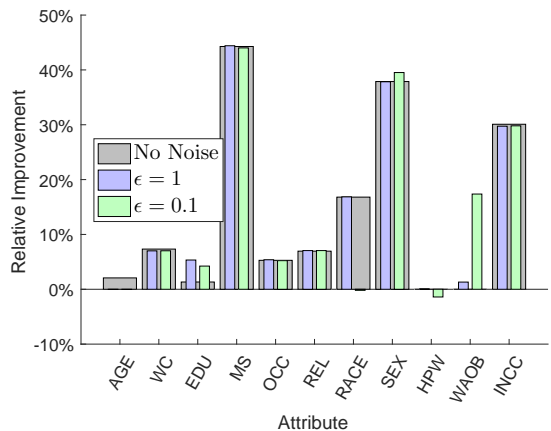}
	\vspace{-6pt}
	\caption{\small Relative Improvement of Model Accuracy of the un-noised, $\varepsilon=1$-DP, and $\varepsilon=0.1$-DP models, with respect to the baseline (marginals).
	Overall, the improvement for $\varepsilon=1$ or $\varepsilon=0.1$ is comparable to that for the un-noised version. Adding noise to achieve DP for structure learning (Section~\ref{sec:gm:structure}) can lead to a different acyclic graph of the model. (This is why there is a significant difference in improvement for attributes RACE and WAOB between $\varepsilon=1$-DP and $\varepsilon=0.1$-DP.)}
	\label{fig:acs_me_relimp}
	\vspace{-8pt}
\end{figure}
\section{Evaluation}\label{sec:evaluation}
We feed the 2013 ACS dataset (Section~\ref{sec:data}) as input to our tool and generate millions of synthetic records. We start with a description of the experimental setup. The evaluation itself is divided into four logical parts: (Section~\ref{subsec:eval_stat}) statistical measures (how good are the synthetics according to well-established statistical metrics); (Section~\ref{subsec:eval_ml}) machine learning measures (how good are the synthetics for machine learning tasks, specifically classification); (Section~\ref{subsec:eval_game}) distinguishing game (how successful is an adversary at distinguishing between a real record and a synthetic one); and (Section~\ref{subsec:eval_perf}) performance measures (how computationally complex it is to generate synthetics).

\subsection{Setup}\label{subsec:evalsetup}
To achieve differential privacy we sampled the input dataset into disjoint sets of records. Each of $\ds{T}$ and $\ds{P}$ contains roughly $280,000$ records, whereas $\ds{S}$ contains roughly $735,000$ records (Section~\ref{sec:gm:dprivacy}). For differential privacy of the generative model, we set $\varepsilon=1$ (though we give some results for $\varepsilon=0.1$) and always set $\delta$ to be at most $2^{-30} \approx 10^{-9}$.

We typically compare the quality of our generated synthetics with real records (coming from the input dataset) and privacy-preserving marginals (Section~\ref{sec:gm:synthesis}) which we refer to as \emph{reals} and \emph{marginals}, respectively. The synthetics we generate are referred by their generation parameters (e.g., $\omega = 10$). Unless otherwise stated, we set $k=50$, $\varepsilon_0=1$, $\gamma = 4$, and $\omega$ is set to vary between $5$ and $11$. 

We maintain a testing set of roughly $100,000$ records. Evaluation of classifiers (in this section) uses at least $100,000$ records for training and a (disjoint) testing set of size at least $30$\% of the size of the aforementioned training set.

\begin{figure}
	\centering
	\includegraphics[width=0.95\linewidth]{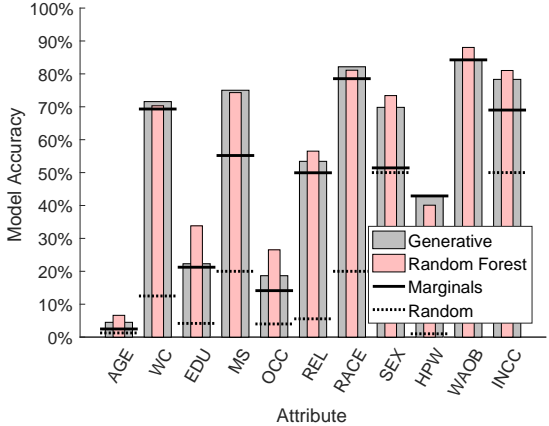}
	\vspace{-6pt}
	\caption{\small Model Accuracy. The difference between the random forest accuracy and the marginals accuracy indicates how informative the data is about each attribute.}
	\label{fig:acs_me_acc}
	\vspace{-12pt}
\end{figure}
\subsection{Statistical Measures}\label{subsec:eval_stat}
We evaluate the quality of the synthetics in terms of their statistical utility, i.e., the extent to which they preserve the statistical properties of the original (input) dataset. We can do this at the level of the generative model (Section~\ref{sec:gm:model}) itself. Concretely, we directly quantify the error of the privacy-preserving generative model before any synthetic record is generated. We do this for each attribute by repeatedly selecting a record from the input dataset (uniformly at random) and using the generative model to find the most likely attribute value (of that attribute) given the other attributes. The generative model error is then measured as the proportion of times that the most likely attribute value is \emph{not} the correct (i.e., original) one. We repeat this procedure millions of times to quantify the average error of the model for each attribute. Because the generative model is made differentially private by adding noise (Section~\ref{sec:gm:parameters_dp}) we additionally repeat the whole procedure $20$ times (learning a different private model each time) and take the average.

The results are shown in Figures~\ref{fig:acs_me_relimp} and~\ref{fig:acs_me_acc}. Figure~\ref{fig:acs_me_relimp} shows the relative decrease in model error (i.e., improvement of model accuracy) over the (privacy-preserving) marginals; it shows this improvement for the un-noised, $(\varepsilon\!\!=\!\!1)$-differential private, and $(\varepsilon\!=\!0.1)$-differential private generative models. There is a clear accuracy improvement over marginals, in addition to a low decrease in improvement between the un-noised model and the $\varepsilon\!=\!1$ and $\varepsilon\!=\!0.1$ noisy versions.  

Figure~\ref{fig:acs_me_acc} shows the accuracy of the un-noised generative model against the (un-noised) marginals, random guessing (baseline), and the best classifier we could find (trained on as many records as the generative model), the random forest (RF). While RF's accuracy is sometimes higher than that of the generative model, the accuracy of the latter is in many cases significantly higher than that of marginals and random guessing.  We conclude that while the proposed generative model does not perform as well as RF (though making RF differentially private would certainly lower its performance) it does perform significantly better than marginals (or random guessing).
\begin{figure}
	\centering
	\includegraphics[width=0.95\linewidth]{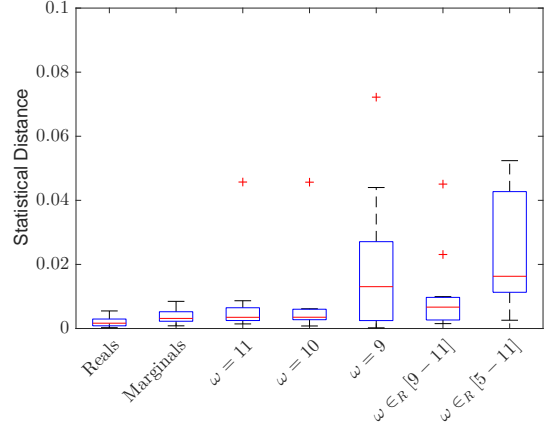}
	\vspace{-6pt}
	\caption{\small Statistical Distance for individual attributes of two distributions: reals and (other) reals; reals and marginals; reals and synthetics (for varying $\omega$). The smaller the statistical distance the more information is preserved. The distance of reals and $\omega = 11$ and $\omega = 10$ synthetics is similar to that of reals and marginals.}
	\label{fig:acs_single_sd}
	\vspace{-8pt}
\end{figure}

In addition to the error of the generative model, we can more directly evaluate the extent to which the generated synthetics preserve the statistical properties of the original (input) dataset. To do this, we compare the probability distributions of the synthetics with the reals and marginals. Specifically, for reals, marginals and synthetics datasets, we compute the distribution of each attribute and of each pair of attributes. We compare each of these distributions to those computed on (other) reals and quantify their distance. We use a well-established statistical distance metric called ``the'' statistical distance (a.k.a. total variation distance~\cite{gibbs2002choosing,levin2009markov}).

The results are shown in Figures~\ref{fig:acs_single_sd} and~\ref{fig:acs_pairs_sd}, where Figure~\ref{fig:acs_single_sd} shows box-and-whisker plots for the distance of the distributions of each attribute separately, and Figure~\ref{fig:acs_pairs_sd} shows box-and-whisker plots for the distance of the distributions of all pairs of attributes.
While marginals do well for single attribute and sometimes outperform our synthetics (though the statistical distance for all datasets is small), synthetics clearly outperform marginals for pairs of attributes. We conclude that the generated synthetics preserve significantly more statistical information than marginals.

\subsection{Machine Learning Measures}\label{subsec:eval_ml}
In addition to preserving statistical properties of the original (input) dataset, the synthetics should also be suitable to various machine learning tasks. In particular, given a learning task, we can evaluate the extent to which synthetics are suitable replacements for a real dataset. For the ACS dataset, a natural and well-establish classification task is to predict a person's income class (i.e., $\geq 50$K or $< 50$K) using the other attributes as features (Section~\ref{sec:data}).
\begin{figure}
	\centering
	\includegraphics[width=0.95\linewidth]{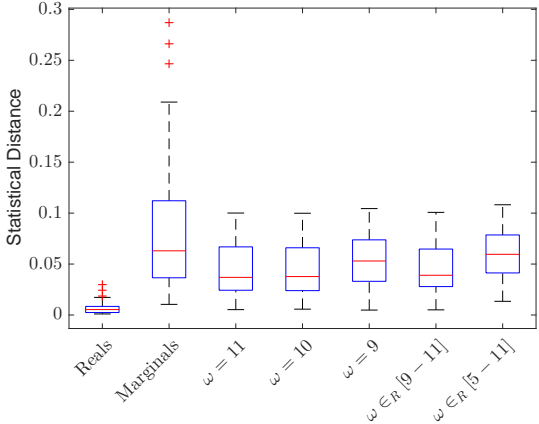}
	\vspace{-6pt}
	\caption{\small Statistical Distance for pairs of attributes of two distributions: reals and (other) reals; reals and marginals; reals and synthetics (for varying $\omega$). The smaller the statistical distance the more information is preserved. The distance of reals and synthetics is significantly smaller than that of reals and marginals.}
	\label{fig:acs_pairs_sd}
	\vspace{-8pt}
\end{figure}

We train various classifiers on the synthetic datasets and on the real (input) dataset. We then compare: the classification accuracy obtained, and the agreement rate of the learned classifiers. Specifically, for two classifiers trained on different datasets (but with the same classification task), we define the \emph{agreement rate} to be the percentage of records for which the two classifiers make the same prediction (regardless of whether the prediction is correct). Given that we look at the agreement rate of classifiers trained on reals and synthetics, the agreement rate reveals the extent to which the classifier trained on synthetic data has learned the \emph{same} model as the classifier trained on real data.

Table~\ref{tbl:acs13_classif_comp} shows the obtained results for three (best) classifiers: Classification Tree (Tree), Random Forest (RF), and AdaBoostM1 (Ada). The accuracy and agreement rate are calculated as the average over $5$ independent runs, that is, for each run, we use different (randomly sampled) training and testing datasets. Overall, we see that both the accuracy and the agreement rates of the synthetics are significantly closer to that of the reals than the marginals are.

In addition to comparing the best classifiers trained on real data versus those trained on synthetic data, we can also compare privacy-preserving classifiers trained on real data versus non-private classifiers trained on (privacy-preserving) synthetic data.  In particular, Chaudhuri et al.~\cite{chaudhuri2011differentially} propose two techniques based on empirical risk minimization to train logistic regression (LR) and support vector machines (SVM) binary classifiers: output perturbation (noise is added to the learned model), and objective perturbation (noise is added to the objective function of the minimization problem). To train such classifiers, we first pre-process our datasets following the instructions in~\cite{chaudhuri2011differentially}: we transform each categorical attribute into an equivalent set of binary attributes, and normalize features so that each feature takes values in $[0,1]$ and subsequently further normalize each training example such that its norm is at most $1$. The target attribute for classification is again the person's income class. The method proposed in~\cite{chaudhuri2011differentially} has two parameters: the privacy budget $\varepsilon$ which we set to $1$ (the same as for our generative model), and $\lambda$ which is a regularization parameter. We use the code of~\cite{chaudhuri2011differentially}, which we obtain courtesy of the authors, to train the LR and SVM classifiers. Because the classification models vary greatly depending on $\lambda$, we vary its value in the set $\{10^{-3},10^{-4},10^{-5}, 10^{-6}\}$ and (optimistically) pick whichever value maximizes the accuracy of the non-private classification model.

We report the accuracy obtained in each case in Table~\ref{tbl:acs13_classif_comp_lr_svm}, where we contrast non-private, output perturbation DP, and objective perturbation DP classifiers trained on real data with non-private classifiers trained on our synthetic datasets (for various values of $\omega$). Remark that for the case $\omega=11$, for example, this is a fair comparison as the obtained LR and SVM classifiers are $\varepsilon=1$-DP and thus provides the exact same privacy guarantee as the output perturbation and objective perturbation LR and SVM classifiers. Non-private LR and SVM classifiers trained on our (privacy-preserving) synthetic datasets are competitive with differentially private LR and SVM classifiers trained on real data.
\begin{table}
    \caption{\small Classifier Comparisons. The agreement rate is the proportion of times that the classifier makes the same prediction as a classifier trained on real data.}
    \vspace{4pt}
	{\scriptsize
	\resizebox{\columnwidth}{!}{%
	    \begin{tabular}{l|l|l|l||l|l|l|}
	    \cline{2-7}
	    ~            & \multicolumn{3}{|c||}{Accuracy} & \multicolumn{3}{|c|}{Agreement Rate} \\
	    \cline{2-7}
	    ~            						& Tree & RF & Ada & Tree & RF & Ada \\ \hline
	    \multicolumn{1}{|l|}{Reals} 		& 77.8\% & 80.4\%  & 79.3\% & 80.2\%  & 86.4\% & 92.4\% \\ \hline
	    \multicolumn{1}{|l|}{Marginals} 	& 57.9\% & 63.8\% & 69.2\% & 58.5\% & 65.4\% & 75.6\% \\ \hline
	    \multicolumn{1}{|l|}{$\omega=11$} 	& 72.4\% & 75.3\% & 78.0\% & 73.9\% & 79.0\% & 83.0\% \\ \hline
	    \multicolumn{1}{|l|}{$\omega=10$}  	& 72.3\% & 75.2\% & 78.1\% & 73.8\% & 78.9\% & 83.6\% \\ \hline
	    \multicolumn{1}{|l|}{$\omega=9$} 	& 72.4\% & 75.2\% & 77.5\% & 73.9\% & 79.2\% & 82.4\%\\ \hline
	    \multicolumn{1}{|l|}{$\omega \in_R [9-11]$} 	& 72.3\% & 75.2\% & 78.1\% & 73.7\% & 79.0\% & 83.9\%\\ \hline
	    \multicolumn{1}{|l|}{$\omega \in_R [5-11]$} 	& 72.1\% & 75.2\% & 78.1\% & 73.6\% & 79.2\% & 83.3\%\\ \hline
	    \end{tabular}
    	}%
    }

    \label{tbl:acs13_classif_comp}
    \vspace{-8pt}
\end{table}

We emphasize that the results should be interpreted in favor of our proposed framework. Indeed, the classifiers trained on our privacy-preserving synthetics outperforms $\varepsilon$-DP LR classifier and only achieves about $1\%$ lower accuracy than the objective-perturbation $\varepsilon$-DP SVM.  This is significant because the technique to train the $\varepsilon$-DP LR and SVM is specifically optimized for that task. In contrast, our synthetics are \emph{not} specifically generated to optimize any particular classification task; instead the general objective is to preserve the statistical properties of real data.
\begin{table}
    \caption{\small Privacy-Preserving Classifier Comparisons.}
    \vspace{4pt}
\begin{center}
	{\scriptsize
	\resizebox{0.75\columnwidth}{!}{%
	    \begin{tabular}{l|c|c|}
	    \cline{2-3}
	    ~            						& LR & SVM \\ \cline{2-3} \bottomrule 
	    \multicolumn{1}{|l|}{Non Private} 	& 79.9\% & 78.5\% \\ \hline
	    \multicolumn{1}{|l|}{Output Perturbation} 	& 69.7\% & 76.2\% \\ \hline
	    \multicolumn{1}{|l|}{Objective Perturbation} 	& 76.3\% & 78.2\% \\ \toprule \bottomrule
	    
	    \multicolumn{1}{|l|}{Marginals} & 68.9\% & 68.9\% \\ \toprule \bottomrule
	    
	    \multicolumn{1}{|l|}{$\omega=11$} 	& 77.6\% & 77.2\% \\ \hline
	    \multicolumn{1}{|l|}{$\omega=10$}  	& 77.7\% & 77.1\% \\ \hline
	    \multicolumn{1}{|l|}{$\omega=9$} 	& 77.5\% & 77.1\% \\ \hline
	    \multicolumn{1}{|l|}{$\omega \in_R [9-11]$} 	& 77.5\% & 76.9\% \\ \hline
	    \multicolumn{1}{|l|}{$\omega \in_R [5-11]$} 	& 77.7\% & 77.3\% \\ \toprule
	    \end{tabular}
    	}%
    }
    \end{center}
    \label{tbl:acs13_classif_comp_lr_svm}
    \vspace{-12pt}
\end{table}

\subsection{Distinguishing Game}\label{subsec:eval_game}
A different way to evaluate the quality of synthetic datasets is to quantify the extent to which the synthetics can ``pass off'' as real records. In other words, we can imagine a game in which the participant is given a random record either from a real dataset or a synthetic dataset (but doesn't know which) and is asked to distinguish between the two possibilities. In this case, the utility is measured by how likely a sophisticated participant (e.g., a well-established learning algorithm) is to make a mistake, i.e., confuse a synthetic record with a real record or vice-versa.

For our purpose the role of the participant is played by the two best classifiers (those that best distinguish synthetics from reals): Random Forest (RF) and Classification Tree (Tree). Specifically, we provide $50,000$ records from both a real dataset and a synthetic dataset (i.e., $100,000$ total) as training examples to the (binary) classifier. We then evaluate the accuracy on a $50$\% mix of real and synthetic records which were not part the training set. Table~\ref{tbl:acs_dfs} shows the results: both classifiers obtain reasonably high (79.8\% and 73.2\%) accuracy in distinguishing marginals from real records. However, both classifiers obtain much lower accuracy (i.e., 63\%) when trying to distinguish synthetics from reals.
\begin{table}
    \caption{\small Distinguishing Game. Random Forest (RF) and Classification Tree (Tree) can easily distinguish marginals from reals but perform significantly less well when trying to distinguish synthetics from reals.}
    \vspace{4pt}
	{\scriptsize
	\resizebox{\columnwidth}{!}{%
	    \begin{tabular}{l|c|c|c|c|c|c|c|}
	    \cline{2-8}
		    ~    & \multirow{2}{*}{Reals} & \multirow{2}{*}{Marginals} & \multicolumn{5}{c|}{$\omega$ $=$ or $\in_R$} \\ \cline{4-8}
		    ~    &   &  & $11$ & $10$ & $9$ & $[9-11]$ & $[5-11]$ \\ \cline{2-8} \bottomrule
		    \multicolumn{1}{|l|}{RF} & 50\% & 79.8\% & 62.3\% & 61.8\% & 63.0\% & 60.1\% & 61.4\% \\ \hline
		    \multicolumn{1}{|l|}{Tree} & 50\% & 73.2\% & 58.9\% & 58.6\% & 59.8\% & 57.9\% & 58.4\% \\ \toprule
		\end{tabular}
    	}%
    }

    \label{tbl:acs_dfs}
    \vspace{-8pt}
\end{table}
\begin{figure}
	\centering
	\includegraphics[width=0.95\linewidth]{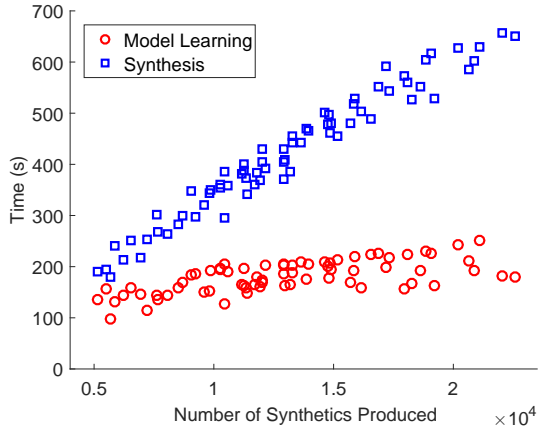}
	\caption{\small Synthetic Generation Performance. The parameters are: $\omega=9$, $k=50$, $\gamma=4$. The time to generate $10,000$ synthetic records on a single-core is less than $10$ minutes. Thus, in the same time frame we can generate $1$ million records with $100$ parallel instances.}
	\label{fig:acs_time_scatter}
	\vspace{-6pt}
\end{figure}
\subsection{Performance Measures}\label{subsec:eval_perf}
In addition to how much utility they preserve, synthetics also need to be easy to generate. The generation is a parallel process, so we measure the time taken for both the learning of the privacy-preserving generative model (model learning) and the synthetics generation (synthesis) itself. Figure~\ref{fig:acs_time_scatter} shows the time taken to produce various number of synthetics records (totaling over $1$ million). The parameters \texttt{max\_plausible} and \texttt{max\_check\_plausible} (Section~\ref{sec:tool}) were set to $100$ and $50000$ respectively. The machine used for the experiment runs Scientific Linux and is equipped with an Intel Xeon E5-2670 processor ($2.60$GHz) with $32$ processing units and $128$GB of RAM. We ran $96$ instances ($16$ in parallel at a time) and picked a random maximum runtime for each instance between $3$ and $15$ minutes.

The generator outputs all synthetics produced regardless of whether they pass the privacy test. Naturally, only those which pass the test should to be released. Thus, the extent to which we can synthesize large (privacy-preserving) datasets depends on how easy it is to find synthetics that pass the privacy-test (Section~\ref{sec:pd}). To evaluate this, we set $\gamma = 2$ and $\texttt{max\_check\_plausible} = 100,000$, and vary $k$ and $\omega$. We measure the proportion of synthetics which pass the privacy test. The results are shown in Figure~\ref{fig:acs_ptc}: even for stringent privacy parameters (e.g., $k=100$) a significant proportion (i.e., over $50$\% for $\omega \in_R [5-11]$) of synthetics pass the test. 

\begin{figure}
	\centering
	\includegraphics[width=0.95\linewidth]{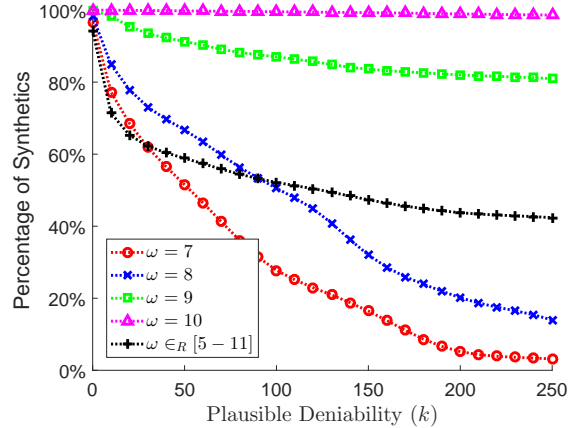}
	\caption{\small Percentage of candidates which pass the privacy test for various values of $k$ and $\omega$ ($\gamma = 2$). The percentage decreases for higher privacy (i.e., larger $k$) but remains significant even for combinations of parameters yielding high privacy. The conclusion is that (very) large synthetic datasets can efficiently be generated.}
	\label{fig:acs_ptc}
	\vspace{-8pt}
\end{figure}
\section{Related Work}\label{sec:relatedwork}
Data synthesis is the process of creating artificial data that can be used to serve some key purposes of real data. For example, it can be used to test software~\cite{whiting2008creating,pedersen2006simple} when there is an inadequate supply of real tests available from operations. It can be used to evaluate the performance of database systems~\cite{gray1994quickly} when obtaining real data at sufficient scale is difficult or expensive.  It can also be used to protect privacy. In 1993, Rubin~\cite{rubin1993statistical} proposed the idea of creating synthetic data based on multiple imputation, that is, on repeated use of a function that proposes values for missing fields in a record. The generative model we presented in Section~\ref{sec:gm:model} uses a similar technique. This and related work have given rise to a substantial body of research on the release of synthetic data~\cite{keller2016does,drechsler2011synthetic}. Such techniques have achieved significant deployment; for example, they have been adopted by the U.S. Census Bureau~\cite{hawala2008producing,kinney2011towards}.

An alternative to data synthesis sometimes called \emph{syntactic} privacy protection transforms the sensitive data using a combination of aggregation, suppression, and generalization, to achieve criteria such as $k$-anonymity~\cite{samarati1998generalizing} or $l$-diversity~\cite{machanavajjhala2007diversity}. Although these techniques support privacy protected data publishing without synthesis, the degree of privacy protection they provide depends on the background knowledge of adversaries. The key difference between $(k,\gamma)$-plausible deniability and $k$-anonymity is that the latter is a syntactic condition on the \emph{output} of a \emph{sanitization process}, whereas plausible deniability is a condition on a \emph{synthetic generator mechanism} with respect to its \emph{input} data.

Plausible deniability as a privacy notion for synthetic data was proposed by Bindschaedler and Shokri in~\cite{bindschaedler2016synthesizing} which describes a technique to synthesize location trajectories in a privacy-preserving way. The use of plausible deniability in~\cite{bindschaedler2016synthesizing} is specific to location privacy as it is defined in terms of a semantic distance between two location trajectories. In contrast, this work generalizes the notion of plausible deniability for \emph{general} data synthesis by establishing it as a privacy criterion in terms of the underlying synthesis probabilities. Consequently, this criterion is applicable to any system and any (kind of) data. The generative framework described in this paper enables us to formally connect plausible deniability to differential privacy (Theorem~\ref{thm:dp}).

Differential privacy provides guarantees even against adversaries with (almost) unlimited background knowledge. However, popular differentially private mechanisms such as the Laplacian mechanism target the release of aggregate statistics. By contrast, we focus on synthesizing data with the \emph{same format} as the (sensitive) input data. Preserving the data format is valuable for many reasons, such as enabling the use of applications and code that are targeted at raw or sanitized data. There is a line of work on mechanisms that are differentially private and provide data as an output. Some of these techniques have theoretical properties that may make them impractical in important cases~\cite{blum2013learning}.

A prominent example is the Exponential Mechanism~\cite{mcsherry2007mechanism}. 
Informally, the mechanism induces a distribution on the space of output records by assigning a weight to each such record and then producing output records by sampling from that distribution. The mechanism is of great importance for algorithm design due to its generality. However, as several researchers have pointed out~\cite{cormode2012differentially,Lantz2015SEM,blocki2016differentially}, a direct application is too costly to be practical for high-dimensional datasets due to the complexity of sampling, which grows exponentially in the dimension of the data records. Concretely, a straightforward implementation of the exponential mechanism to generate synthetic records from the ACS dataset (Section~\ref{sec:data}) would sample from a universe of records of size roughly $2^{39}$ (Table~\ref{tbl:acs13stats}). This would require pre-computing that many probabilities. If we assume we can store each value in four bytes this would require about $2$TB of memory. In contrast, the complexity of synthesizing a record with our framework depends only on the number of records in the dataset and on the complexity of our generative model and thus the process is very efficient in practice (Section~\ref{subsec:eval_perf}).

There is a growing collection of mechanisms and case studies for differentially private release of data~\cite{abowd2008protective,charest2011can,mcclure2012differential,machanavajjhala2008privacy,wasserman2010statistical,chen2011publishing,jagannathan2008privacy}, although some of these are based on a broad view of data release, such as the release of histograms or contingency tables. Our use of plausible deniability to achieve differentially private data adds to this body of work. The typical approach to protect privacy in this context is to add noise directly to the generative model. For example, this is the approach taken by~\cite{li2014differentially,bowen2016differentially,liu2016model,zhang2014privbayes}. In particular,~\cite{zhang2014privbayes} constructs a generative model based on Bayesian networks similar to the generic generative model of Section~\ref{sec:gm}.

Our work takes a novel approach: instead of trying to achieve differential privacy directly, we design the generative framework to achieve plausible deniability. A major step towards achieving plausible deniability and a key novelty is the idea of \emph{testing} privacy. That is, instead of designing the mechanism so it achieves that notion by design, we use a privacy test which rejects ``bad'' samples. As a side effect, the generative model is decoupled from the framework. Privacy is guaranteed in a way that is oblivious to the specifics of the generative model used for synthesis. Furthermore, the guarantee offered by our proposed $(k,\gamma)$-plausibly deniable mechanism is surprisingly close to that of differential privacy, as evidenced by the fact that merely randomizing the threshold yields a differentially private mechanism.

The idea of running data synthesis and then testing privacy has been used before. For example, Reiter et al. in~\cite{reiter2009estimating} and~\cite{reiter2014bayesian} use inference to evaluate privacy risk for a synthetic data release. However, there is no proof of privacy, and it is not efficient to run a set of inference attacks to estimate the risk before releasing a dataset. 
\section{Discussion}\label{sec:discuss}

Regardless of whether one intends to release a set of aggregate statistics or a synthetic dataset, there is no privacy protection technique that can preserve utility for all meaningful utility functions. However, one key feature of our generative framework is that, unlike other approaches based on differential privacy, any generative model $\mathcal{M}$ can be used while keeping the same privacy guarantees. 
As a result, designing $\mathcal{M}$ is a data science problem which need not involve considerations of the privacy of the seeds.

Parameter $\omega$ (Section~\ref{sec:gm}) controls the closeness of synthetics to their seeds. (Lower $\omega$ means more dependence on the seed but it is harder to pass the privacy test.) A pragmatic approach is to generate synthetics for various values of $\omega$ and then randomly sample a subset of those synthetics which pass the privacy test (this is evaluated in Section~\ref{sec:evaluation}). Note that no matter what the value of $\omega$ is, the privacy of the seeds is ensured because of the privacy test.

In the special case where the generative model $\mathcal{M}$ is independent of the seed, the privacy guarantee applies to any output from Mechanism~\ref{def:mechanism} because the privacy test always passes. However, for a seed-dependent generative model, the privacy of the seeds is safeguarded by rejecting synthetics which do not pass the privacy test. So, when generating several synthetics using the same input dataset, the privacy obtained depends on the number of synthetics released. 
In fact, when Privacy Test~\ref{def:randprivacytest} is used, the $(\varepsilon,\delta)$-differential privacy guarantee applies to a single (released) synthetic record $y$. That said, the composition theorems for differential privacy can be used to extend the guarantee to arbitrarily large synthetic datasets, provided the privacy budget is appropriately increased. We leave as future work the design of improved composition strategies.

\section{Conclusions}
We have presented the first practical data synthesis tool with strong privacy guarantees. We have formalized plausible deniability for generating generic data records, and have proven that our mechanisms can achieve differential privacy without significantly sacrificing data utility. 

{
    \medskip \noindent \textbf{Acknowledgments.}
	This work was supported in part by NSF CNS grants 13-30491 and 14-08944. The views expressed are those of the authors only.
}

{	\scriptsize
	\bibliographystyle{abbrv}
	\bibliography{refs}
}

\appendix
\section{Composing Differential Privacy}\label{app:dp}

\begin{theorem}[Sequential Composition -- 3.16~\cite{dwork2014algorithmic}]
	Let $\mathcal{F}_i$ be an $(\varepsilon_i,\delta_i)$-differentially private mechanism, for $i=1,\ldots,m$. Then, for any dataset $D$, the mechanism which releases outputs: $\left(\mathcal{F}_1(D), \mathcal{F}_2(D), \ldots, \mathcal{F}_m(D) \right)$ is an $(\varepsilon,\delta)$-differentially private mechanism for $\varepsilon = \sum_{i=1}^{m} \varepsilon_i$, and $\delta = \sum_{i=1}^{m} \delta_i$.
\label{thm:dp:seqc}
\end{theorem}

\begin{theorem}[Advanced Composition -- 3.20~\cite{dwork2014algorithmic}]
	Let $\mathcal{F}_i$ be an $(\varepsilon,\delta)$-differentially private mechanism, for $i=1,\ldots,m$. Then the mechanism represented by the sequence of $k$ queries over $\mathcal{F}_1(\cdot), \mathcal{F}_2(\cdot), \ldots, \mathcal{F}_m(\cdot)$ \emph{with potentially different inputs} is $(\varepsilon',\delta')$-differentially private for:
	\begin{eqnarray*}
	\varepsilon' = \varepsilon \sqrt{2 k \ln{\frac{1}{\delta''}}} + k \varepsilon ( e^\varepsilon - 1 ) \quad \text{ and }& \delta' = k\delta + \delta'' \ .
	\end{eqnarray*}
\label{thm:dp:advc}
\end{theorem}

\begin{theorem}[Amplification of Sampling --~\cite{li2012sampling}]\ \\
	Let mechanism $\mathcal{F}$ be an $(\varepsilon,\delta)$-differentially private mechanism. The mechanism which first sub-samples each record of its input dataset with probability $\delta < p < 1$ and then runs $\mathcal{F}$ on the sub-sampled dataset is $(\varepsilon',\delta')$-differentially private for:
	\begin{eqnarray*}
	\varepsilon' = \ln{\left(1 + p (e^\varepsilon - 1)\right)} \quad \text{ and }& \delta' = p\delta \ .
	\end{eqnarray*}
\label{thm:dp:amplif}
\end{theorem}

\section{Proof: Sensitivity of Entropy}\label{app:proofs}
\begin{lemma}[Sensitivity of $\mathrm{H}$]\ \\
Let $\rv{z}$ be a discrete random variable with a probability distribution estimated from $n$ data records. The sensitivity of $\mathrm{H}(\rv{z})$ is:
	\begin{align*}
		\Delta_{H} \leq \frac{1}{n} [2 + \frac{1}{\ln(2)} + 2 \log_2{n}]  \ .
	\end{align*}
	\label{lem:hsensitivity}
\end{lemma}

Lemma~\ref{lem:hsensitivity} is used in Section~\ref{sec:gm:structure_dp}~\eqref{eq:gm:entropy_sensitivity}.

\begin{proof}
	Let $\rv{z}$ and $\rv{z'}$ be the random variables associated with two histograms (of dimension $m$) computed from two neighboring datasets $D$ and $D'$, respectively. Both datasets have $n$ records but differ in exactly one record.
	
	Let $z_c = (c_1, c_2, \ldots, c_m)$ and $z'_c = (c'_1, c'_2, \ldots, c'_m)$ represent the histograms over the $m$ values of the attribute for $\rv{z}$ and $\rv{z'}$, respectively. The entropy is computed over the probability distribution represented by a histogram.
	
	Remark that the histograms of the considered neighboring datasets $D$ and $D'$ can only differ in two positions. If they do not differ in any position, then the Lemma trivially holds ($\Delta_{H} = 0$). That is, without loss of generality, there exists $j_1$ and $j_2$ (with $j_1 \neq j_2$) such that $c'_{j_1} = c_{j_1} + 1$ and $c'_{j_2} = c_{j_2} - 1$. Also, $n-1 \geq c_{j_1} \geq 0$ which means that $n \geq c'_{j_1} \geq 1$, and $n \geq c_{j_2} \geq 1$ which means that $n-1 \geq c'_{j_2} \geq 0$. Furthermore for $i \neq j_1, j_2$, we have $c'_{i} = c_i$, and also: \[ \sum_{i=1}^{m} c_i = \sum_{i=1}^{m} c'_i = n \ . \]

	Now:
	\begin{align*}
		\mathrm{H}(\rv{z}) &= -\sum_{i=1}^{m} \frac{c_i}{n} \log_2{\frac{c_i}{n}} 						\\
					  &= -\frac{1}{n} \left[ \sum_{i=1}^{m} c_i \log_2{c_i} - n \log_2{n} \right]	\\
					  &= \log_2{n} - \frac{1}{n} \left( c_{j_1} \log_2{c_{j_1}} + c_{j_2} \log_2{c_{j_2}} + \sum_{\mathclap{i \neq j_1,j_2}} c_i \log_2{c_i} \right) \ .
	\end{align*}
	Similarly,
	\begin{align*}
		\mathrm{H}(\rv{z'}) &= \log_2{n} - \frac{1}{n}  \sum_{\mathclap{i \neq j_1,j_2}} c_i \log_2{c_i} \\	
		&- \frac{1}{n} \left[ (c_{j_1}+1) \log_2{(c_{j_1}+1)} + (c_{j_2}-1) \log_2{(c_{j_2}-1)} \right] \ .
		\end{align*}
		
	We have that $\Delta_{H} = \max_{c_{j_1},c_{j_2}} \left\lvert \mathrm{H}(\rv{z}) - \mathrm{H}(\rv{z'}) \right\rvert$, but for brevity we omit the max and analyze this quantity with respect to the values of $c_{j_1}$ and $c_{j_2}$ to show that the lemma holds in each case.
	
	Observe that: 
	\begin{align*}
		\Delta_{H} &= \left\lvert \mathrm{H}(\rv{z}) - \mathrm{H}(\rv{z'}) \right\rvert \\
				   &= \frac{1}{n} \left\lvert c_{j_1} \log_2{c_{j_1}} - (c_{j_1}+1) \log_2{(c_{j_1}+1)} \right. \\
				   &+ \left. c_{j_2} \log_2{c_{j_2}} - (c_{j_2}-1) \log_2{(c_{j_2}-1)} \right\rvert \ .
	\end{align*}
	\begin{itemize}
	\item{Case 1: $c_{j_1} = 0$. We have \[ \Delta_{H} = \frac{1}{n} | c_{j_2} \log_2{c_{j_2}} - (c_{j_2}-1) \log_2{(c_{j_2}-1)}| \ . \]
	Clearly, if $c_{j_2} = 1$ then $\Delta_{H} = 0$ and the Lemma trivially holds. So assume $c_{j_2} > 1$.
	We have: 
	\begin{align*}
			\Delta_{H} &= \frac{1}{n} | c_{j_2} \log_2{c_{j_2}} - (c_{j_2}-1) \log_2{(c_{j_2}-1)} | \\
			&= \frac{1}{n} \left\lvert c_{j_2} \log_2{\left(\frac{c_{j_2}}{c_{j_2} - 1}\right)} + \log_2{(c_{j_2} - 1)} \right\rvert \\
			 &\leq \frac{1}{n} \left\lvert c_{j_2} \log_2{\left(\frac{c_{j_2}}{c_{j_2} - 1}\right)} \right\rvert + \frac{1}{n} \log_2{(c_{j_2} - 1)} \\
			 &\leq \frac{1}{n} \log_2{(n-1)} + \frac{1}{n} \left\lvert (a+1) \log_2{\left(1+\frac{1}{a}\right)} \right\rvert \ ,
	\end{align*}
	where $a = c_{j_2}-1 \geq 1$. It is easy to see that $(a+1) \log_2(1+\frac{1}{a}) \leq 2$. We conclude that: $\Delta_{H} \leq \frac{1}{n} (2 + \log_2{n})$.
	}
	\item{Case 2: $c_{j_2} = 1$. We have \[ \Delta_{H} = \frac{1}{n} | c_{j_1} \log_2{c_{j_1}} - (c_{j_1}+1) \log_2{(c_{j_1}+1)}| \ . \]
	Again if $c_{j_1} = 0$, then the Lemma trivially holds. So assume $c_{j_1} > 0$.
	We have:
	\begin{align*}
			\Delta_{H} &= \frac{1}{n} | c_{j_1} \log_2{c_{j_1}} - (c_{j_1}+1) \log_2{(c_{j_1}+1)}| \\
					   &= \frac{1}{n} \left\lvert  c_{j_1} \log_2{(\frac{c_{j_1}}{c_{j_1}+1})} - \log_2{(c_{j_1}+1)} \right\rvert \\
					   &\leq \frac{\log_2{n}}{n} +  \frac{1}{n} \left\lvert  c_{j_1} \log_2{(\frac{c_{j_1}}{c_{j_1}+1})} \right\rvert \\
					   &= \frac{\log_2{n}}{n} + \frac{1}{n} c_{j_1} \log_2{(\frac{c_{j_1}+1}{c_{j_1}})} \\
					   &= \frac{\log_2{n}}{n} + \frac{1}{n} c_{j_1} \log_2{\big(1 + \frac{1}{c_{j_1}}\big)} \ .
	\end{align*}
	Using L'Hopital's rule we have $ c_{j_1} \log_2{\big(1 + \frac{1}{c_{j_1}}\big)} \leq \frac{1}{\ln{2}}$. We conclude that: $\Delta_{H} \leq \frac{1}{n} (\frac{1}{\ln{2}} + \log_2{n})$.
	}
	\item{Case 3: $c_{j_1} \geq 1$, $c_{j_2} \geq 2$.
		We have: 
	\begin{align*}
			\Delta_{H} &= \frac{1}{n} | c_{j_1} \log_2{c_{j_1}} - (c_{j_1}+1) \log_2{(c_{j_1}+1)} \\
					   &+ c_{j_2} \log_2{c_{j_2}} - (c_{j_2}-1) \log_2{(c_{j_2}-1)} | \\
					   &\leq \frac{1}{n} | c_{j_1} \log_2{c_{j_1}} - (c_{j_1}+1) \log_2{(c_{j_1}+1)}| \\
					   &+ \frac{1}{n} | c_{j_2} \log_2{c_{j_2}} - (c_{j_2}-1) \log_2{(c_{j_2}-1)} | \ ,
	\end{align*}
	where it is seen that the two terms have been bounded for cases 1 and 2. Thus, putting it all together, we conclude that $\Delta_{H} \leq \frac{1}{n} \left( 2 + \frac{1}{\ln{2}} + 2 \log_2{n} \right)$.
	}
	\end{itemize}
\end{proof}

\section{Proof: Connection With Differential Privacy}\label{app:connection}
In this section, we prove Theorem~\ref{thm:dp} of Section~\ref{sec:pd}.

\begin{theoremrec}[Differential Privacy of $\mathcal{F}$]\ \\
Let $\mathcal{F}$ denote Mechanism~\ref{def:mechanism} with the (randomized) Privacy Test~\ref{def:randprivacytest} and parameters $k \geq 1$, $\gamma > 1$, and $\varepsilon_0 > 0$. For any neighboring datasets $D$ and $D'$ such that $|D|, |D'| \geq k$, any set of outcomes $Y \subseteq \mathcal{U}$, and any integer $1 \leq t < k$, we have:
	\begin{align*}
		\pr{\mathcal{F}(D') \in Y} \leq e^{\varepsilon} \pr{\mathcal{F}(D) \in Y} + \delta \ ,
	\end{align*}
	for $\delta = e^{-\varepsilon_0 (k-t)}$ and $\varepsilon = \varepsilon_0 + \ln{(1+\frac{\gamma}{t})}$.
\end{theoremrec}

We start with some notation. Let $\mathcal{U}$ denote the universe of data records. All data records, i.e., those from the datasets $D$ and $D'$, including synthetic records produced by $\mathcal{M}$ (and $\mathcal{F}$) are elements of $\mathcal{U}$. Let $D$ and $D'$ denote two neighboring datasets, i.e., either $D = D' \cup \{ d \}$ for some $d \in \mathcal{U}$, or $D' = D \cup \{ d' \}$ for some $d' \in \mathcal{U}$. We assume that both $D$ and $D'$ have at least $k$ records, and we have parameters $k \geq 1$, $\gamma > 1$, $\varepsilon_0 > 0$. For convenience we write $p_d(y) = \prsmall{y = \mathcal{M}(d)}$, and refer to $\mathcal{M}$ only implicitly.

Given a dataset $D^\star$, we want to reason about the probability that synthetic record $y$ is released: $\pr{\mathcal{F}(D^\star) = y}$. Observe that given synthetic record $y \in \mathcal{U}$, the records of $D^\star$ can be partitioned (into disjoint sets) by the privacy criterion. Concretely, let $I_d(y)$ be the partition number of a record $d \in D^\star$ with respect to $y$. The \emph{partition number} $I_d(y)$ is the unique non-negative integer such that $\gamma^{-(I_d(y)+1)} < p_d(y) \leq \gamma^{-I_d(y)}$. In other words, $I_d(y) = \lfloor - \log_{\gamma} p_d(y) \rfloor$. If $p_d(y) = 0$ then the partition number is undefined. Similarly, we define the \emph{partition} (or partition set) for $i \geq 0$  as $C_i(D^\star,y) = \{ d : d \in D^\star, I_d(y) = i \}$. That is, partition $i$ is the set of records with partition number $i$. 

A key step is to express the probability $\pr{\mathcal{F}(D^\star) = y}$ in terms of: (1) the probability of generating $y$ from a specific partition (i.e., the seed is in that partition) and (2) the probability of passing the test. For (2) remark that the probability of passing the privacy test depends only on the partition of the seed (see Privacy Test~\ref{def:randprivacytest}).
\begin{definition}
	For any dataset $D^{\star}$, if the seed is in partition $i$, the probability of passing the privacy test is given by: ${\rm{pt}}(D^\star, i, y) = \pr{L \geq k - |C_i(D^\star,y)|}$, where $L \sim \rm{Lap}(\frac{1}{\varepsilon_0})$.
\end{definition}
\begin{definition}
	For any dataset $D^{\star}$, the probability of producing $y$ from partition $i$ is: 
		\[ q(D^\star, i, y) = {\rm{pt}}(D^\star, i, y) \sum_{s \in C_i(D^\star,y)} p_s(y) \ . \]
\end{definition}
As the following demonstrates, $\pr{\mathcal{F}(D^\star) = y}$ is readily expressed in terms of (1) and (2).
\begin{lemma}
	\label{lem:fdstarexpr}
	\label{lem:decomp}	
	For any dataset $D^{\star}$ and any synthetic record $y \in \mathcal{U}$ we have: 
	\begin{align}  
	\pr{\mathcal{F}(D^{\star}) = y} = \frac{1}{|D^{\star}|} \sum_{i \geq 0} q(D^\star, i, y)  \ .
	\end{align}
	
\end{lemma}

In other words, the probability of releasing $y$ (from $D^\star$) can be expressed as the sum, over all partitions, of the probability of generating $y$ from a given partition and then releasing it.
\begin{proof}[of Lemma~\ref{lem:decomp}]
	Fixing a $y$ and following the description of Mechanism~\ref{def:mechanism}, we have:
		\begin{align*}
			\pr{\mathcal{F}(D^\star) = y} &= \sum_{s \in D^\star} \pr{s \text{ is seed}, \mathcal{F}(D^\star) = y} \\		
			&= \frac{1}{|D^\star|} \sum_{s \in D^\star} p_s(y) \ \pr{(D^\star,s,y) \text{ passes test}} \ ,
		\end{align*}
	given the fact that the seed is sampled uniformly at random.

	Further, observe that terms of the sum for which $p_s(y) = 0$ can be omitted and that we can partition the set $\{ d : d \in D^\star, p_d(y) > 0 \}$ with respect to the partition number of its elements. That is:
	\[ \{ d : d \in D^\star, p_d(y) > 0 \} = \cup_{i \geq 0} C_i(D^\star, y) \ , \]
	where for each $d \in D^\star$ such that $p_d(y) > 0$, there exists a unique non-negative integer $j$ such that $d \in C_j(D^\star, y)$.
	
	Thus:
		\begin{align*}
			&\pr{\mathcal{F}(D^\star) = y} \\
			&= \frac{1}{|D^\star|} \sum_{s \in D^\star} p_s(y) \ \pr{(D^\star,s,y) \text{ passes test}} \\
			&= \frac{1}{|D^\star|} \sum_{i \geq 0} \sum_{s \in C_i(D^\star, y)} \!\!\!\! p_s(y) \ \pr{(D^\star,s,y) \text{ passes test}} \\
			&= \frac{1}{|D^\star|} \sum_{i \geq 0} \ \pr{L \geq k - |C_i(D^\star,y)|} \!\!\!\! \sum_{s \in C_i(D^\star, y)} \!\!\!\! p_s(y) \\
			&= \frac{1}{|D^\star|} \sum_{i \geq 0} \ {\rm{pt}}(D^\star, i, y) \!\!\!\! \sum_{s \in C_i(D^\star, y)} \!\!\!\! p_s(y) \\
			&= \frac{1}{|D^\star|} \sum_{i \geq 0} q(D^\star, i, y) \ ,
		\end{align*}
	given that the privacy test depends only on the partition of the seed (and not on the seed itself). (See
    the description of Privacy Test~\ref{def:randprivacytest} with $L$ being drawn from $\rm{Lap}(\frac{1}{\varepsilon_0})$.)
\end{proof}

Remark that $q(D^\star, i, y) = 0$ if and only if $C_i(D^\star, y) = \emptyset$. Also, observe that if a record is added to or subtracted from $D^\star$ then only one partition changes. As a result, we can analyze case-by-case the change in the probability of releasing $y$ from partition $i$, when adding or removing a record to partition $i$. 

The following shows that adding a record to some partition only increases the probability of passing the privacy test by at most $e^{\varepsilon_0}$. (This is a consequence of adding Laplacian noise to the threshold.)
\begin{lemma}
	Given any $y \in \mathcal{U}$, any two neighboring datasets $D$ and $D'$ such that $D' = D \cup \{ d' \}$. For any partition $i$ we have:
	\begin{align*}
	 	{\rm{pt}}(D, i, y) \leq {\rm{pt}}(D', i, y) \leq e^{\varepsilon_0} {\rm{pt}}(D, i, y) \ . 
	\end{align*}
	\label{lem:ptratio}
\end{lemma}

To prove Lemma~\ref{lem:ptratio}, we use the following observation (which comes from the CDF of the Laplace distribution).
\begin{observation}
For any $x \in \mathbb{R}$, if $L$ is a Laplace random variable with shape parameter $b$ and mean $0$, then we have: 
\begin{align*}
	\pr{L \geq x} \leq \pr{L \geq x - 1} \leq e^{\frac{1}{b}} \pr{L \geq x} \ .
\end{align*}
\label{obs:lap-test-bound}
\end{observation}

\begin{proof}[of Lemma~\ref{lem:ptratio}]
	There are two cases: $i = I_{d'}(y)$ or $i \neq I_{d'}(y)$.
	If $i = I_{d'}(y)$ then $d'$ falls into partition $i$ and so $C_i(D',y) = C_i(D,y) \cup \{ d' \}$. We have:
	\begin{align*}
		{\rm{pt}}(D', i, y) &= \pr{L \geq k - |C_i(D',y)|} \\
		&\leq e^{\varepsilon_0} \pr{L \geq k - |C_i(D',y)| + 1} \\
		&= e^{\varepsilon_0} \pr{L \geq k - |C_i(D,y)|} = e^{\varepsilon_0} {\rm{pt}}(D, i, y) \ ,
	\end{align*}
	Also, we have that: ${\rm{pt}}(D', i, y) > {\rm{pt}}(D, i, y)$.
	
	Otherwise, if $i \neq I_{d'}(y)$ then $C_i(D',y) = C_i(D,y)$, and so:
	\begin{align*}
		{\rm{pt}}(D', i, y) 
		&= \pr{L \geq k - |C_i(D,y)|} = {\rm{pt}}(D, i, y) \ .
	\end{align*}
	Putting it together yields the result.
\end{proof}

To quantify the change in $q(D^\star, i, y)$ due to adding a record to partition $i$ we need to separate two cases: (1) the partition is initially empty (or more generally has initially less than $t$ records) and (2) the partition is not empty (or more generally has at least $t$ records).
\begin{lemma}
	For any $y \in \mathcal{U}$ and any dataset $D$. Let $D' = D \cup \{ d' \}$ for some $d' \in \mathcal{U}$. Let $j$ be the partition number of $d'$ (i.e., $I_d(y) = j$). The following holds.	
	
	\begin{enumerate}
		\item[(a)]%
			For all $i \neq j$, we have $q(D',i,y) = q(D,i,y)$.
		
		\item[(b)]%
			If $|C_j(D, y)| < t$: %
			\[ q(D, j, y) < q(D', j, y) \ , \] and
			\begin{align*}
				q(D', j, y) \leq e^{-\varepsilon_0 (k - t)} \!\!\!\!\!\! \sum_{s \in C_j(D',y)} \!\!\!\! p_s(y) \leq t \ e^{-\varepsilon_0 (k - t)} \ .
			\end{align*}
			If $|C_j(D, y)| \geq t$:
			\begin{align*}
				\frac{q(D', j, y)}{q(D, j, y)} \leq e^{\varepsilon_0} \left[ 1 + \frac{\gamma}{t} \right]  \ .
			\end{align*}	
	\end{enumerate}
	\label{lem:cbycan}
\end{lemma}

\begin{corollary}[of Lemma~\ref{lem:cbycan}]
	For any $y \in \mathcal{U}$ and any dataset $D$. Let $D' = D \cup \{ d' \}$ for some $d' \in \mathcal{U}$. We have $q(D, i, y) \leq q(D', i, y)$, for all $i \geq 0$.
	\label{cor:cbycan}
\end{corollary}

\begin{proof}[of Lemma~\ref{lem:cbycan}]
	Fix $y$ and let $j$ be the partition that $d'$ falls into. 
	
	For part (a), remark that for $i\neq j$, we have $C_i(D,y) = C_i(D',y)$. So $q(D, i, y) = q(D', i, y)$. 
	
	For part (b), we have that:
	\begin{align*}
		q(D, j, y) &= {\rm{pt}}(D, j, y) \!\!\!\!\! \sum_{s \in C_j(D,y)} \!\!\!\! p_s(y) \\
				   &< {\rm{pt}}(D, j, y) \left[ \sum_{s \in C_j(D,y)} \!\!\!\! p_s(y) + p_{d'}(y) \right]  \\
				   &= {\rm{pt}}(D, j, y) \!\!\!\!\!  \sum_{s \in C_j(D',y)} \!\!\!\! p_s(y) \\
				   &\leq {\rm{pt}}(D', j, y) \!\!\!\!\!  \sum_{s \in C_j(D',y)} \!\!\!\! p_s(y) \\
				   &= q(D', j, y) \ ,
	\end{align*}
	given that $p_{d'}(y) > 0$ and ${\rm{pt}}(D', j, y) \geq {\rm{pt}}(D, j, y)$ (Lemma~\ref{lem:ptratio}).
	
	Now, if $|C_j(D, y)| < t$, then:
	\begin{align*}
		q(D', j, y) &= {\rm{pt}}(D', j, y) \!\!\!\!\! \sum_{s \in C_j(D',y)} \!\!\!\! p_s(y) \\
					&\leq e^{-\varepsilon_0 (k - t)} \!\!\!\!\! \sum_{s \in C_j(D',y)} \!\!\!\! p_s(y) \\
					&\leq t \ e^{-\varepsilon_0 (k - t)} \ ,
	\end{align*}
	given that $|C_j(D', y)| \leq t$ and $p_d(y) \leq 1$ for any $d$. Here, the first inequality follows from the fact that ${\rm{pt}}(D', j, y) = \pr{L \geq k - |C_j(D',y)|} \leq \pr{L \geq k - t} = \frac{1}{2} e^{-\varepsilon_0 (k - t)}$.
	
	If $|C_j(D, y)| \geq t$, we have:
	\begin{align*}
		q(D', j, y) &= {\rm{pt}}(D', j, y) \ \!\!\!\!\! \sum_{s \in C_j(D',y)} \!\!\!\! p_s(y) \\
					&= {\rm{pt}}(D', j, y) \ [ \!\!\!\!\!  \sum_{s \in C_j(D,y)} \!\!\!\! p_s(y) + p_{d'}(y) ] \\
					&\leq e^{\varepsilon_0} \ {\rm{pt}}(D, j, y) \ [ \!\!\!\!\!  \sum_{s \in C_j(D,y)} \!\!\!\! p_s(y) + p_{d'}(y) ] \\
					&\leq e^{\varepsilon_0} \ [1 + \frac{\gamma}{t}] \ {\rm{pt}}(D, j, y) \!\!\!\!\!  \sum_{s \in C_j(D,y)} \!\!\!\! p_s(y) \\
					&= e^{\varepsilon_0}  \ [1 + \frac{\gamma}{t}] \ q(D, j, y) \ ,
	\end{align*}
	given Lemma~\ref{lem:ptratio} and the fact that $p_{d'}(y) \leq \gamma \ p_s(y)$ for any $s \in C_j(D,y)$ and so $p_{d'}(y) \leq \frac{\gamma}{t} \sum_{s \in C_j(D,y)}  p_s(y)$.
\end{proof}
The following Lemma is the core result underlying Theorem~\ref{thm:dp}.
\begin{lemma}
Let $\mathcal{F}$ denote Mechanism~\ref{def:mechanism} with the (randomized) Privacy Test~\ref{def:randprivacytest} and parameters $k \geq 1$, $\gamma > 1$, and $\varepsilon_0 > 0$. Take any dataset $D$ with $|D| \geq k$ and let $D' = D \cup \{d'\}$ for any $d' \in \mathcal{U}$. Then for any integer $1 \leq t < k$ and synthetic record $y \in \mathcal{U}$, we have:
	\begin{align*}
		\pr{\mathcal{F}(D) = y} \leq e^{\varepsilon} \pr{\mathcal{F}(D') = y} \ ,
	\end{align*}
	and
	\begin{align*}
		\pr{\mathcal{F}(D') = y} \leq e^{\varepsilon} \pr{\mathcal{F}(D) = y} + \delta \ ,
	\end{align*}
	where $\delta = \delta(D', d', y) \leq e^{-\varepsilon_0 (k-t)}$ and $\varepsilon = \varepsilon_0 + \ln{(1+\frac{\gamma}{t})}$. Here, $\delta(D',d',y) = e^{-\varepsilon_0 (k - t)} |D'|^{-1}  \sum_{s \in C_j(D',y)} p_s(y)$, with $j = I_{d'}(y)$.
	\label{lem:dpcore}
\end{lemma}

\begin{proof}[of Lemma~\ref{lem:dpcore}]
	Fix an arbitrary synthetic record $y \in \mathcal{U}$ and an arbitrary dataset $D$ with $|D| \geq k$. Let $D' = D \cup \{ d' \}$ for some arbitrary $d' \in \mathcal{U}$. Applying Lemma~\ref{lem:decomp} to $D$ we have: $\pr{\mathcal{F}(D) = y} = \frac{1}{|D|} \sum_{i \geq 0} q(D, i, y)$.
	Also, from Corollary~\ref{cor:cbycan} we have $q(D, i, y) \leq q(D', i, y)$ for all $i$. Thus: 
	\begin{align*}
		\pr{\mathcal{F}(D) = y} &= \frac{1}{|D|} \sum_{i \geq 0} q(D, i, y) \\
		&\leq \frac{1}{|D|} \sum_{i \geq 0} q(D', i, y) \\
		&= \frac{|D'|}{|D|} \pr{\mathcal{F}(D') = y} \\
		&\leq \left( 1 + \frac{1}{k} \right) \ \pr{\mathcal{F}(D') = y} \ .	
	\end{align*}
	Observe that since (by assumption) $\gamma > 1$ and $1 \leq t \leq k$, we have: $\frac{1}{k} \leq \frac{1}{t}$, and so $1+\frac{1}{k} < 1 + \frac{\gamma}{t} \leq e^{\varepsilon_0} (1 + \frac{\gamma}{t}) = e^{\varepsilon}$. This shows the first part.
	
	For the second part, apply Lemma~\ref{lem:decomp} to $D'$, and let $j$ be the partition number of $d'$, i.e., $j = I_{d'}(y)$. We have:
	\begin{align*}
		\pr{\mathcal{F}(D') = y} &= \frac{1}{|D'|} \sum_{i \geq 0} q(D', i, y) \\
		&= \frac{1}{|D'|} \left[ \sum_{i \geq 0 : i \neq j} q(D', i, y) + q(D', j, y) \right ] \\
		&= \frac{1}{|D'|} \sum_{i \geq 0 : i \neq j} q(D, i, y) + \frac{q(D', j, y)}{|D'|} \ .
	\end{align*}
	The last equality follows from Lemma~\ref{lem:cbycan} part (a).
	
	Applying Lemma~\ref{lem:cbycan} part (b), we obtain two cases.
	\begin{itemize}
		\item{Case 1: $|C_j(D, y)| < t$. 
			We have:
			\begin{align*}
				\pr{\mathcal{F}(D') = y} &= \frac{1}{|D'|} \sum_{i \geq 0 : i \neq j} q(D, i, y) + \frac{q(D', j, y)}{|D'|} \\
				&\leq \frac{1}{|D'|} \sum_{i \geq 0 : i \neq j} q(D, i, y) + \delta(D',j,y)  \\
				&\leq \frac{1}{|D'|} \sum_{i \geq 0} q(D, i, y) + \delta(D',j,y) \\
				&= \frac{|D|}{|D'|} \pr{\mathcal{F}(D) = y} + \delta(D',j,y) \\
				&\leq \pr{\mathcal{F}(D) = y} + \delta  \ ,
			\end{align*} 
			where $\delta(D',j,y) = \frac{1}{|D'|} e^{-\varepsilon_0 (k - t)} \sum_{s \in C_j(D',y)} p_s(y)$.
		}
		\item{Case 2: $|C_j(D, y)| \geq t$. 
			We have:
			\begin{align*}
				&\pr{\mathcal{F}(D') = y} \\
				&= \frac{1}{|D'|} \left[ \sum_{i \geq 0 : i \neq j} q(D, i, y) + q(D', j, y) \right] \\
				&\leq \frac{1}{|D'|} \left[ \sum_{i \geq 0 : i \neq j} q(D, i, y) +  e^{\varepsilon_0}[1 + \frac{\gamma}{t}] q(D, j, y) \right] \\
				&\leq e^{\varepsilon_0}[1 + \frac{\gamma}{t}] \frac{1}{|D'|} \sum_{i \geq 0} q(D, i, y) \\
				&= e^{\varepsilon_0}[1 + \frac{\gamma}{t}] \frac{|D|}{|D'|} \pr{\mathcal{F}(D) = y} \\
				&\leq e^{\varepsilon} \pr{\mathcal{F}(D) = y} \ ,
			\end{align*}
			for $\varepsilon = \varepsilon_0 + \ln{(1 + \frac{\gamma}{t})}$ given that $\frac{|D|}{|D'|} < 1$.
		}
	\end{itemize}
	
	Letting $\delta(D',d',y) = \delta(D',I_{d'}(y),y)$ finishes the proof.
\end{proof}

With this, we are in a position to prove Theorem~\ref{thm:dpconn}.
\begin{proof}[of Theorem~\ref{thm:dpconn}]
	Fix dataset $D$ with $|D| \geq k$ and any record $d' \in \mathcal{U}$. Let $D' = D \cup \{d'\}$. The range of $\mathcal{F}$ is $\mathcal{U}$ and so any outcome $Y$ is a non-empty subset of $\mathcal{U}$. Fix an arbitrary $Y \subseteq \mathcal{U}$ with $Y \neq \emptyset$. 
	
	We will show that $\pr{F(D_1) \in Y} \leq e^{\varepsilon} \pr{F(D_2) \in Y} + \delta$, whether $D_1 = D$ and $D_2 = D'$, or $D_1 = D'$ and $D_2 = D$.
	
	Consider first the case $D_1 = D$ and $D_2 = D'$. Applying Lemma~\ref{lem:dpcore}, we obtain:
	\begin{align*} 
		\pr{\mathcal{F}(D) \in Y} &= \sum_{y \in Y} \pr{\mathcal{F}(D) = y} \\
		&\leq \sum_{y \in Y} e^\varepsilon \pr{\mathcal{F}(D') = y} \\ 
		&= e^\varepsilon \pr{\mathcal{F}(D') \in Y} \ .
	\end{align*}
	
	Now, consider the case $D_1 = D'$ and $D_2 = D$. Define $c(d',y) = |C_{I_{d'}(y)}(D, y)|$. Given $d'$, we can partition $Y$ between those $y \in Y$ such that the partition in which $d'$ falls has at least $t$ and those such that the partition has less than $t$. That is, $Y = Y_{t-} \cup Y_{t+}$, with $Y_{t-} = \{ y : y \in Y, c(d',y) < t \}$ and $Y_{t+} = \{ y : y \in Y, c(d',y) \geq t \}$. We have:
	\begin{align*} 
		\pr{\mathcal{F}(D') \in Y} &= \sum_{y \in Y} \pr{\mathcal{F}(D') = y} \\
		&\leq  \sum_{y \in Y_{t+}}  e^\varepsilon \pr{\mathcal{F}(D) = y} \\
		&+ \sum_{y \in Y_{t-}}  [ e^\varepsilon \pr{\mathcal{F}(D) = y} + \delta(D', I_{d'}(y), y) ] \\
		&= e^\varepsilon \sum_{y \in Y} \pr{\mathcal{F}(D) = y} +  \sum_{y \in Y_{t-}}  \delta(D', d', y) \\
		&= e^\varepsilon \pr{\mathcal{F}(D) \in Y} + \sum_{y \in Y_{t-}} \delta(D', d', y) \ ,
	\end{align*}
	where the inequality applies cases of Lemmas~\ref{lem:cbycan} and~\ref{lem:dpcore} separately to each $y$ depending on whether $y \in Y_{t-}$ (case 1 in the proof of Lemma~\ref{lem:dpcore}) and $y \in Y_{t+}$ (case 2 in the proof of Lemma~\ref{lem:dpcore}).
	
	It remains to show that $\sum_{y \in Y_{t-}} \delta(D', d', y) \leq e^{-\varepsilon_0 (k - t)}$. For this define $C(D', d', y) = C_{I_{d'}(y)}(D', y)$. We have:
	\begin{align*}
		\sum_{y \in Y_{t-}} \delta(D', d', y) &= \sum_{y \in Y_{t-}} \!\!\! \frac{e^{-\varepsilon_0 (k - t)}}{|D'|} \!\!\!\! \sum_{s \in C(D', d', y)} \!\!\!\!\!\!\! p_s(y) \\
		&= \frac{e^{-\varepsilon_0 (k - t)}}{|D'|} \sum_{s \in D'} \sum_{y \in Y_{t-}} \mathbbm{1}_{I_{d'}(y) = I_s(y)} \ p_s(y) \\
		&\leq e^{-\varepsilon_0 (k - t)} \ ,
	\end{align*}
	given that $\sum_{y \in Y_{t-}} \mathbbm{1}_{I_{d'}(y) = I_s(y)} \  p_s(y) \leq \sum_{y \in \mathcal{U}} p_s(y) \leq 1$ for all $s \in D'$.
	
\end{proof}

\end{document}